\documentclass[11pt]{article}
\usepackage[utf8]{inputenc}
\usepackage[dvipsnames]{xcolor}
\usepackage{placeins}
\usepackage{multirow}
\usepackage[margin=1.125in]{geometry}
\usepackage[dvipsnames]{xcolor}
\usepackage[normalem]{ulem}
\definecolor{DarkGreen}{rgb}{0.1,0.5,0.1}
\definecolor{DarkRed}{rgb}{0.5,0.1,0.1}
\definecolor{DarkBlue}{rgb}{0.1,0.1,0.5}

\usepackage{mleftright}

\usepackage[small]{caption}
\usepackage{subcaption}
\usepackage[pdftex]{hyperref}
\hypersetup{
    unicode=false,          
    pdftoolbar=true,        
    pdfmenubar=true,        
    pdffitwindow=false,      
    pdfnewwindow=true,      
    colorlinks=true,       
    linkcolor=DarkBlue,          
    citecolor=DarkGreen,        
    filecolor=DarkGreen,      
    urlcolor=DarkBlue,          
    pdftitle={The list size of random linear and LDPC codes at capacity},
    pdfauthor={Shashwat Silas},
}
\usepackage{amssymb,amsmath,amsthm,amsfonts,enumitem, dsfont}
\usepackage{fullpage,comment,mathtools}
\usepackage{mathrsfs}
\usepackage{tikz}
\usepackage{pgfplots}
\pgfplotsset{compat=1.14}
\usetikzlibrary{arrows,decorations.pathreplacing,decorations.pathmorphing,decorations.shapes,snakes,patterns,positioning}
\usepackage[ruled,linesnumbered]{algorithm2e}
\usepackage{thm-restate}
\mathtoolsset{showonlyrefs}

\newcommand{\A}{\ensuremath{\mathcal{A}}}

\newcommand{\cC}{\ensuremath{\mathcal{C}}}

\newcommand{\cB}{\ensuremath{\mathcal{B}}}
\newcommand{\cH}{\ensuremath{\mathcal{H}}}
\newcommand{\cM}{\ensuremath{\mathcal{M}}}

\newcommand{\cP}{\ensuremath{\mathcal{P}}}

\newcommand{\cR}{\ensuremath{\mathcal{R}}}
\newcommand{\cT}{\ensuremath{\mathcal{T}}}

\newcommand{\R}{{\mathbb R}}
\newcommand{\Z}{{\mathbb Z}}

\newcommand{\F}{{\mathbb F}}

\newcommand{\CC}{{\mathbb C}}

\newcommand{\PR}[1]{\mathrm{Pr}\left[ #1\right]}
\newcommand{\PROver}[2]{\mathrm{Pr}_{#1}\left[ #2\right]}

\newcommand{\E}[1]{\mathbb{E}\left[ #1\right]}
\newcommand{\EE}{\mathbb{E}}
\newcommand{\Eover}[2]{\mathop{\mathbb{E}}_{#1}\left[ #2 \right]}

\newcommand{\ceil}[1]{\left\lceil {#1}\right\rceil}
\newcommand{\floor}[1]{\left\lfloor {#1}\right\rfloor}

\newcommand{\inabs}[1]{\left|#1\right|}
\newcommand{\ip}[2]{\ensuremath{\left\langle #1,#2\right\rangle}}

\newcommand{\inset}[1]{\left\{#1\right\}}
\newcommand{\fracpart}[1]{\left\{#1\right\}}                     

\newcommand{\inparen}[1]{\left(#1\right)}
\newcommand{\inbrak}[1]{\left[#1\right]}
\newcommand{\suchthat}{\,:\,}

\newcommand{\supp}{\mathrm{supp}}

\newcommand{\poly}{\mathrm{poly}}

\newcommand{\spn}{\ensuremath{\operatorname{span}}}
\newcommand{\dist}{\ensuremath{\operatorname{dist}}}
\newcommand{\diag}{\ensuremath{\operatorname{diag}}}

\newcommand{\tr}{\mathrm{tr}}

\newcommand{\TRLC}{R_{\mathrm{RLC}}}

\newcommand{\Epl}{\widetilde{\mathbb{E}}}

\newcommand{\eps}{\varepsilon}
\renewcommand{\epsilon}{\varepsilon}

\definecolor{brightpink}{rgb}{1.0, 0.0, 0.5}
\definecolor{byzantine}{rgb}{0.74, 0.2, 0.64}
\definecolor{byzantium}{rgb}{0.44, 0.16, 0.39}

\DeclareMathOperator*\Ber{Bernoulli}

\newcommand\one{\mathds{1}}

\newcommand{\rows}[1]{\mathrm{rows}\mleft(#1\mright)}

\newcommand{\onesvector}{\mathbf{1}}

\renewcommand{\exp}[1]{\mathrm{exp}\mleft(#1\mright)}

\theoremstyle{plain}
\declaretheorem[name=Theorem,numberwithin=section]{theorem}
\declaretheorem[name=Lemma,sibling=theorem]{lemma}

\newtheorem*{lemma*}{Lemma}
\newtheorem*{theorem*}{Theorem}
\newtheorem{definition}[theorem]{Definition}

\newtheorem{conjecture}[theorem]{Conjecture}

\newtheorem{corollary}[theorem]{Corollary}

\newtheorem{claim}[theorem]{Claim}
\newtheorem{proposition}[theorem]{Proposition}

\newtheorem{fact}[theorem]{Fact}

\newtheorem{question}{Question}[section]

\newcommand{\Rs}[1]{R^*_{#1}}                 
\newcommand{\LowW}[1]{\ensuremath{\mathrm{Low}_{<#1}}}             
\newcommand{\Dq}[2]{\ensuremath{D_q\mleft(#1\,\middle\|\,#2\mright)}} 
\newcommand{\cF}{\ensuremath{\mathcal{F}}}

\title{The list size of random linear and LDPC codes at capacity}
\date{}

\author{Shashwat Silas\\[2pt] \small Independent researcher\\ \small\texttt{shashwat@alumni.stanford.edu}}

\begin{document}
\maketitle
\thispagestyle{empty}

\begin{abstract}
\noindent
We show that for every prime power $q$ and $\rho\in(0,1-\frac1q)$ there exists $\varepsilon_0(q, \rho) > 0$ such that for all positive $\varepsilon < \varepsilon_0(q, \rho)$, a random linear code over $\mathbb{F}_q$ of rate $1-h_q(\rho)-\varepsilon$ has list size $\lfloor h_q(\rho)/\varepsilon\rfloor$ or $\lfloor h_q(\rho)/\varepsilon\rfloor+1$ with probability $1-o(1)$. Before this work, the list size was known only up to a constant factor for $q\ge3$, and for $q=2$ it was known to be either $\lfloor h_2(\rho)/\varepsilon\rfloor-1$, $\lfloor h_2(\rho)/\varepsilon\rfloor$ or $\lfloor h_2(\rho)/\varepsilon\rfloor+1$. We also show that the same two-point concentration holds for the list size of random LDPC codes. This follows from a general theorem establishing that random LDPC codes and random linear codes have the same threshold rate for distance, list decoding, list recovery from errors and any other \emph{convex} property, which answers a question of Mosheiff, Resch, Ron-Zewi, Silas and Wootters.
\end{abstract}
\newpage

\section{Introduction}\label{sec:intro}
A code is list-decodable if every small Hamming ball contains only a few of its codewords. Formally, a code $\cC\le\F_q^n$ is \emph{$(\rho,L)$-list-decodable} if
\begin{equation}\label{eq:ld}
\left|\inset{c\in\cC\suchthat\dist(c,y)\le\rho n}\right|\;\le\;L
\qquad\text{for any }y\in\F_q^n,
\end{equation}
where $\dist$ is the Hamming distance. The
\emph{rate} of $\cC$ is $\log_q|\cC|/n$. The largest rate at which there are codes with $L$
independent of $n$ is $1-h_q(\rho)$, the
\emph{list-decoding capacity}, where $h_q$ is the $q$-ary entropy
function. Fix a prime power $q$ and a radius
$\rho\in\inparen{0,1-\frac1q}$. Let $\cC\le\F_q^n$ be a \emph{random
	linear code of rate $R = 1-h_q(\rho)-\eps$}, i.e.\ the image of a uniformly random matrix in $\F_q^{n\times Rn}$. We refer to $\eps$ as the
\emph{gap to capacity}, and to the least $L$ such that
$\cC$ is $(\rho,L)$-list-decodable as the \emph{list size} $L^*(\cC,\rho)$ of $\cC$.

\begin{question}\label{q:list-size}
What is the list size of a random linear code?
\end{question}
Zyablov and Pinsker \cite{ZP81} first studied this question almost fifty years ago, and there is a natural answer that we would expect from a simple first-moment argument. The list size is larger than $L$ exactly when some ball of relative radius $\rho$ contains $m = L+1$ codewords $c^{(1)},\dots,c^{(m)}$. We estimate the expected number of such lists in $\cC$. Since the code is linear, we may assume without loss of generality that $c^{(1)} = 0$. Then the center of the ball has relative distance at most $\rho$ from $0$, and the list consists of $0$ and the $L$ differences $c^{(i)}-c^{(1)}$. Recall that a Hamming ball of relative radius $\rho$ in $\F_q^n$ has size $q^{nh_q(\rho)+o(n)}$. In other words, a uniformly random vector in the ball has entropy $h_q(\rho)$ on average over its coordinates. There are about $q^{nh_q(\rho)}$ choices of the center and, for each of them, about $q^{nh_q(\rho)}$ choices of each difference. Hence there are about $q^{nmh_q(\rho)}$ choices of the center and the differences together. If the $L$ differences are linearly independent, a random linear code of rate $R$ contains all of them with probability $q^{-n(1-R)L}$. Thus the expected number of such lists in $\cC$ is about
\[
q^{n\inparen{mh_q(\rho)-(1-R)L}+o(n)}\;=\;q^{n\inparen{h_q(\rho)-\eps L}+o(n)},
\]
which is exponentially small when $L>h_q(\rho)/\eps$ and exponentially large when $L<h_q(\rho)/\eps$. So from the first-moment argument we would expect that $L^* = \ceil{h_q(\rho)/\eps}$. From the same argument we could also guess the \emph{threshold rate} for $(\rho,L)$-list-decodability, that is, the rate below which a random linear code is $(\rho,L)$-list-decodable with probability $1-o(1)$, and above which it is not. Namely, the expected number of such lists grows exactly when
\[
R\;>\;1-\frac{(L+1)\,h_q(\rho)}{L}\;=\;1-h_q(\rho)\inparen{1+\frac1L}.
\]

However, it is not at all clear that these guesses would be correct. Indeed, a large expectation does not mean that a list exists with high probability, and of course, in the argument above we only considered lists whose differences are linearly independent. If the span of the differences has dimension $d<L$, the code contains them with the much larger probability $q^{-n(1-R)d}$. While the number of such lists is smaller, the differences can be dependent in many ways, and for each of these ways we have a different expected number of such lists.

After a long line of work to determine the list size (discussed in Section~\ref{sec:related}), surprisingly, the answer turns out to be the same as the one we expected from the first-moment argument. More precisely, we show that there is $\eps_0 = \eps_0(q,\rho)>0$ such that for all $\eps<\eps_0$, with probability $1-o(1)$, the list size is exactly $\ceil{h_q(\rho)/\eps}$ if the fractional part of $h_q(\rho)/\eps$ is not exponentially small in $1/\eps$, and otherwise it is either $\floor{h_q(\rho)/\eps}$ or $\floor{h_q(\rho)/\eps}+1$. This two-point concentration of the list size is reminiscent of the well-known two-point concentration of the clique number and the chromatic number of random graphs \cite{BE76,AN05}.

In fact, the same two-point concentration holds for random LDPC codes, where our result follows from a much more general theorem: random LDPC codes and random linear codes have the same threshold rate for distance, list decoding, average-radius list decoding, list recovery from errors and any other \emph{convex} property (Section~\ref{sec:intro-ldpc}). This answers a question of \cite{MRRSW}, and shows that random LDPC codes are combinatorially equivalent to random linear codes when it comes to many natural properties. Our contributions are as follows.
\begin{enumerate}
\item[(1)] \textbf{The list size of random linear codes at capacity.} For small $\eps$ we prove a much stronger form of the conjecture in \cite{GLMRSW}, who suggested that the true list size is $(1+o(1))\,h_q(\rho)/\eps$. Along the way, we also show that the list size is $h_q(\rho)/\eps + O_{q,\rho}(1)$ for all positive $\eps<1-h_q(\rho)$. 
\item[(2)] \textbf{The threshold rate for list decoding.} We show that the threshold rate of random linear codes for $(\rho,L)$-list-decodability is $1-h_q(\rho)\inparen{1+\frac1L}$ (up to a $q^{-\Omega(L)}$ error), as expected from the first-moment argument above. Since $h_q(\rho) \to \rho$ as $q\to \infty$, this ties together nicely with the recent results for Reed--Solomon codes over large alphabets, which showed that their threshold rate is the generalized Singleton bound, $1-\rho\inparen{1+\frac1L}$ \cite{AGL24,LMS25}.
\item[(3)] \textbf{A bound on entropy under sparse linear constraints.} As our main technical tool for the lower bound in~(2), we show that if the coordinates of a random vector $X\in\F_q^k$ and some linear combinations of two or more of them are each nonzero with probability at most $\rho$, then $H_q(X)\le k\,h_q(\rho)-\gamma^*$ for a constant $\gamma^*>0$. Further, if there are many such combinations, then this bound tightens to nearly $(k-1)\,h_q(\rho)$ (Theorem~\ref{thm:tool}). The proof uses the Gibbs variational principle and the Erd\H{o}s--Rado sunflower lemma \cite{ER60}, and may be of independent interest. This result also implies a generalization of a central lemma in \cite{GHK11} (Corollary~\ref{cor:hamming}).
\item[(4)] \textbf{Random LDPC codes.} We show that random LDPC codes and random linear codes have the same threshold rate for every convex property by proving a converse of \cite[Theorem~1.9]{MRRSW}. This proof uses the second moment method and the small subgraph conditioning theorem of Janson \cite{Jan95}. In particular,~(1) and~(2) hold for random LDPC codes.
\end{enumerate}

\subsection{Random linear codes}\label{sec:contributions}

\begin{restatable}{theorem}{MainTwoPoint}\label{thm:main-twopoint}
Let $q$ be a prime power and $\rho\in\inparen{0,1-\frac1q}$. There are
$\eps_0 = \eps_0(q,\rho)>0$ and $c = c(q,\rho)>0$ such that the following holds for every fixed $0<\eps<\eps_0$.
Let $\cC\le\F_q^n$ be a uniformly random linear code of rate $R = 1-h_q(\rho)-\eps$,
and let $\zeta:=h_q(\rho)/\eps$. With probability $1-o(1)$ as $n\to\infty$, if the fractional part of $\zeta$ is larger than $q^{-c/\eps}$ then
\[
L^*(\cC,\rho)\;=\;\ceil{\zeta},
\]
and otherwise $L^*(\cC,\rho)\in\inset{\floor \zeta,\ \floor \zeta+1}$.
\end{restatable}

For every $q$ our lower bound is one greater than that of \cite{GLMRSW}, and for $q\ge3$ our upper bound is $\floor{h_q(\rho)/\eps}+1$, which improves the best known $O(1/\eps)$ bound of \cite{GHK11}.

\begin{corollary}\label{cor:all-eps} For any fixed positive $\eps<1-h_q(\rho)$, with probability $1-o(1)$ a random linear code $\cC$ of rate $1-h_q(\rho)-\eps$ satisfies $L^*(\cC,\rho)\;\le\;\zeta+O_{q,\rho}(1)$.
\end{corollary}

To prove Theorem~\ref{thm:main-twopoint}, we bound the threshold rate for $(\rho,L)$-list-decodability, which we denote by $\TRLC(\rho,L)$, using the threshold framework of \cite{MRRSW}. \cite{MRRSW} defined the concept of a \emph{local} property, which is one characterized by the exclusion of a family of \emph{bad} matrices invariant under permutations of rows and of columns. The elements of any set of $m = L+1$ codewords in some Hamming ball of relative radius $\rho$ can be viewed as the columns of a matrix in $\F_q^{n\times m}$, and a code $\cC$ is $(\rho,L)$-list-decodable if and only if no such matrix is contained in $\cC$. As in \cite{MRRSW}, we refer to the row distribution of a matrix as its \emph{type}, denoted by $\tau$, and to the types of bad matrices as \emph{bad types}. The type $\tau$ of such a matrix is a distribution on $\F_q^m$, and we let $r:=\dim(\tau)$ be the dimension of the span of its support.

For a single type, we know from \cite[Section~2]{MRRSW} that the number of matrices of that type with $n$ rows is $q^{nH_q(\tau)+O(\log n)}$, where $H_q$ is the base-$q$ entropy. And a random linear code of rate $R$ contains a fixed one of these matrices with probability $q^{-n(1-R)r}$. So in expectation, the code contains roughly $q^{n\inparen{H_q(\tau)-(1-R)r}}$ of them, and in particular, this expectation decays for $R$ below $1-H_q(\tau)/r$.

However, as \cite{MRRSW} observe, $1-H_q(\tau)/r$ is not necessarily the threshold rate. Indeed, for any linear map $A$, if a code contains any matrix $M$ of type $\tau$, then by linearity it also contains the matrix $MA^{\top}$, whose type they call `$\tau$-implied'. And they show that a matrix of type $\tau$ appears in the code only if the rate is large enough so that each of its implied types is also expected to appear in it. This is analogous to random graphs, where $G(n,p)$ is likely to contain a copy of a fixed graph $H$ once the expected number of copies of each subgraph of $H$ tends to infinity. So after accounting for implied types, they establish the threshold rate for a random linear code to contain a matrix of type $\tau$ as
\[
\Rs\tau\;=\;\max_{\tau'}\inparen{1-\frac{H_q(\tau')}{\dim(\tau')}},
\]
where the maximum is over the $\tau$-implied types $\tau'$ of positive dimension. Thus $\TRLC(\rho,L)$ is the infimum of $\Rs\tau$ over the \emph{$\rho$-clustered} types $\tau$ of length $m$, that is, the types of $m$ codewords in some Hamming ball of relative radius $\rho$. We prove this in Lemma~\ref{fact:threshold}. The list size is then
\begin{equation}\label{eq:implicit}
L^*(\cC,\rho) \;=\; \min\inset{L\ge1\suchthat
\TRLC(\rho,L)>R}
\end{equation}
with probability $1-o(1)$.

\begin{restatable}{theorem}{MainUpper}\label{thm:main-upper}
Let $q$ be a prime power and $\rho\in\inparen{0,1-\frac1q}$. There is
$L_0 = L_0(q,\rho)$ such
that for
every integer $L\ge L_0$,
\[
1-h_q(\rho)\inparen{1+\frac 1L}
\;\le\;\TRLC(\rho,L) \;\le\; 1-h_q(\rho)\inparen{1+\frac1L}+q^{-\Omega_{q,\rho}(L)}.
\]
\end{restatable}

By \eqref{eq:implicit}, for $\eps$ small enough, Theorem~\ref{thm:main-upper} implies that $L^*\le\floor\zeta+1$, and if the fractional part of $\zeta$ is not exponentially small in $1/\eps$, that $L^*\ge\ceil\zeta$.

\subsection{Random LDPC codes}\label{sec:intro-ldpc}

LDPC codes, introduced by Gallager \cite{Gal62}, are linear codes defined by sparse parity checks. Because of their sparsity, LDPC codes have efficient decoding algorithms and have been of deep interest to theorists and practitioners alike.

We study the ensemble of random $s$-LDPC codes of \cite{MRRSW}, the Gallager ensemble generalized to larger alphabets with random coefficients. For an odd integer $s$, the parity-check matrix of a random $s$-LDPC code of rate $R$ has $t:=(1-R)s$ independent layers. In each layer, the $n$ coordinates are partitioned by a uniformly random permutation into $n/s$ blocks of size $s$, and each block $b$ corresponds to the parity check $\sum_{i\in b}\xi_ix_i = 0$, where the coefficients $\xi_i$ are independent and uniform in $\F_q^\times$. \cite{MRRSW} showed that these random LDPC codes achieve list-decoding capacity and, in fact, achieve every local property that random linear codes achieve. However, they observed that random LDPC codes and random linear codes do not have the same local properties.  

 For example, for $q=2$ every codeword of a random LDPC code has even weight, because the parity checks of each layer sum to the all-ones check. The property of only containing codewords of even weight is a local property, and its bad matrices are the single columns of odd weight. Random LDPC codes of every rate have this property, while random linear codes of nonzero rate do not. \cite[Remark~1.11]{MRRSW} asked the following question.

\begin{question}\label{q:ldpc}
Is there a natural class of local properties for which random LDPC codes and random linear codes have the same threshold?
\end{question}

We show that random LDPC codes and random linear codes have the same threshold rate for a broad class of local properties, which we call \emph{convex}. Informally, a local property is convex if a mixture of two of its bad types is again a bad type. We say that a distribution $\nu$ on $\F_q^m$ is \emph{scale invariant} if $\nu(\xi v) = \nu(v)$ for any $\xi\in\F_q^\times$.

\begin{theorem}\label{thm:ldpc-same} Let $\cP$ be a local property whose bad matrices have $m$ columns, and let $\cT$ be a set of distributions on $\F_q^m$ such that, for random linear codes, the threshold rate for $\cP$ is $\TRLC:= \inf_{\tau\in\cT}\Rs\tau$. Suppose that $\cP$ is convex, that is,
\begin{enumerate}[label=\textup{(\arabic*)}]
\item $\cT$ is closed under mixtures and under multiplication of rows by elements of $\F_q^\times$,
\item for any $\tau\in\cT$ some $\tau'\in\cT$ has support equal to the span of $\supp\tau$, and
\item every scale-invariant $\nu\in\cT$ whose support is a subspace has scale-invariant \emph{roundings} $\nu_n$ to denominator $n$ (distributions with probabilities in $\frac1n\Z$ and $\nu_n-\nu = O(1/n)$) that are bad types of $\cP$ with support in $\supp\nu$, and, when $q=2$, whose matrices have columns of even weight.
\end{enumerate}
Then for any $\gamma>0$ there is $s_0$ such that for any odd $s\ge s_0$, a random $s$-LDPC code of rate $R$ satisfies $\cP$ with probability $1-o(1)$ if $R\le\TRLC-\gamma$, and does not satisfy $\cP$ with probability $1-o(1)$ if $\TRLC+\gamma\le R\le1-\gamma$.
\end{theorem}

Distance, list decoding, average-radius list decoding \cite{GN14,GLMRSW,GLS26} and list recovery from errors \cite{RW18,GLMRSW} are all convex, while the even-weight property does not satisfy condition~(3). By \cite[Theorem~1.9]{MRRSW}, a random LDPC code of rate below the threshold of random linear codes for a local property satisfies the property with probability $1-o(1)$. For the converse, we show that for any scale-invariant type $\tau$ with full support, random LDPC codes of rate above $\Rs\tau$ are likely to contain a matrix of type $\tau_n$, a rounding of $\tau$ to denominator $n$ (Theorem~\ref{thm:ldpc-core}). Scale-invariant types are natural here, since multiplying a coordinate of a random LDPC code by an element of $\F_q^\times$ does not change the distribution of the code. By conditions~(1)--(3), there is such a type whose threshold rate is close to $\TRLC$ and whose roundings are bad types (Section~\ref{sec:ldpc-single}). For example, by condition~(1), the average of a type $\tau\in\cT$ over the multiplications of its rows by elements of $\F_q^\times$ is in $\cT$. And by concavity of entropy, its threshold rate is at most $\Rs\tau$. Since list decoding is convex, our results on list decoding hold for random LDPC codes as well.

\begin{corollary}\label{cor:ldpc-list}
Let $q$ be a prime power and $\rho\in\inparen{0,1-\frac1q}$. For any $L\ge1$ and $\gamma>0$ there is $s_0$ such that for any odd $s\ge s_0$, the threshold rate of random $s$-LDPC codes for $(\rho,L)$-list-decodability is $\gamma$-close to $\TRLC(\rho,L)$ in the sense of Theorem~\ref{thm:ldpc-same}. In particular, for any fixed $\eps$ and $s$ large enough, Corollary~\ref{cor:all-eps} holds for random $s$-LDPC codes, and so does Theorem~\ref{thm:main-twopoint} if $\eps<\eps_0$.
\end{corollary}

\subsection{Related work}\label{sec:related}
\paragraph{Bounds on the list size of random linear codes.} Zyablov and Pinsker \cite{ZP81} first showed that random linear codes achieve capacity with $L = q^{O(1/\eps)}$. Guruswami, H\r{a}stad and Kopparty \cite{GHK11} improved this to $L = O(1/\eps)$, and Guruswami and Narayanan \cite{GN14} showed an $\Omega(1/\eps)$ lower bound, together establishing that the list size is $\Theta_{q,\rho}(1/\eps)$. For $q=2$, Li and Wootters \cite{LW18} showed, using a potential-function argument, that the list size is at most $h_2(\rho)/\eps+O(1)$. Later, using the threshold framework of \cite{MRRSW} for random linear codes, Guruswami, Li, Mosheiff, Resch, Silas and Wootters \cite{GLMRSW} proved that the leading constant is at least $h_q(\rho)$ for any $q$. For $q=2$, it follows from their bounds that the list size is $\floor\zeta-1$, $\floor\zeta$ or $\floor\zeta+1$.

In concurrent work, Yuan and Zhu \cite{YZ26} prove $L^*\le\ceil\zeta+O_{q,\rho}(1)$ for any positive $\eps<1-h_q(\rho)$, where the additive constant is exponential in $\kappa^{-2}$, for $\kappa := 1-\frac{q\rho}{q-1}$. This result is equivalent to our observation in Corollary~\ref{cor:all-eps}. Additionally, Kumar and Mon \cite{KM26} prove that for any $\rho\in\inparen{0,1-\frac1q}$ and any positive $\eps<1-h_q(\rho)$, a random linear code of rate $1-h_q(\rho)-\eps$ is $\inparen{\rho,16q(\zeta+1)}$-list-decodable. Both works use techniques different from ours: \cite{YZ26} directly generalize the arguments of \cite{GHK11}, and \cite{KM26} develop a new polynomial method for Furstenberg sets.

\paragraph{Local properties and thresholds.} As discussed above, \cite{MRRSW} developed the threshold framework for local properties of random linear codes and showed that random LDPC codes achieve any local property that a random linear code does. They left open the question of whether these codes have the same threshold rate for some natural class of properties, which we resolve here. Their threshold framework has been extended to random codes \cite{GMRSW}, to punctured low-bias codes \cite{GM24}, and, using local coordinate-wise linear properties, to random Reed--Solomon codes \cite{LMS25}. Over large alphabets, the threshold rate of random Reed--Solomon codes for list decoding is the generalized Singleton bound \cite{AGL24,LMS25}, which has the same form as our result.

Guruswami, Mosheiff, Resch, Silas and Wootters \cite{GMRSW} studied the \emph{symmetric} properties of random codes, whose bad types approach a convex set, and found the threshold rate of every symmetric property using the average of a bad type over permutations of its columns. Our convex properties are similar to their symmetric properties, but in our proof we use the average of a bad type over multiplication of its rows by nonzero elements of $\F_q$, which is again a bad type by condition~(1) of Theorem~\ref{thm:ldpc-same}.

\subsection{Proof overview}\label{sec:overview}

\subsubsection{Random linear codes}\label{sec:overview-rlc}

\paragraph{Upper bound on the threshold rate.} Recall that $\TRLC(\rho,L)$ is the infimum of $\Rs\tau$ over the types $\tau$ of lists of $L+1$ codewords in some Hamming ball of relative radius $\rho$. So, to show an upper bound on $\TRLC(\rho,L)$ it suffices to exhibit a single bad type for list-decodability which is likely to appear in the code.  We tighten the analysis of the bad type introduced in \cite[Definition~4.2]{GLMRSW}, which is the distribution of $(u_1+\alpha,\dots,u_m+\alpha)$, where the coordinates of $u\in\F_q^m$ are i.i.d. random variables that are nonzero with probability exactly $\rho$ and uniform on $\F_q\setminus\inset0$ when nonzero, and $\alpha$, the center, is uniform on $\F_q$. \cite[Lemma~4.4]{GLMRSW} estimate $\Rs\tau$ for this type $\tau$, and we calculate $\Rs\tau$ exactly. In this type the only dependence between the codewords is the common center $\alpha$, and so the maximum over the implied types is attained at the implied type modulo the span of the all-ones vector. We prove this in Section~\ref{sec:lower}, and show that for this type the first-moment argument is tight up to $\Delta_m\le q^{-\Omega(m)}$, the entropy of $u$ conditioned on $u\bmod\F_q\onesvector$. This entropy is small, since $u$ is unlikely to be confused with $u+c\onesvector$ for $c\neq0$. Indeed, in each coordinate the distributions of $u_i$ and $u_i+c$ have Bhattacharyya coefficient $\lambda<1$, and in Lemma~\ref{lem:delta} we show that $\Delta_m = O(\lambda^m)$. Thus we obtain the upper bound in Theorem~\ref{thm:main-upper}.  

\paragraph{Lower bound on the threshold rate.} To lower bound  the threshold rate we must show that \emph{no} bad type is likely to appear in the code, and this is the harder direction. We achieve this by first reducing the lower bound on $\Rs\tau$ to a bound on the entropy of a random vector with sparse coordinates under sparse linear constraints,  which we then prove as a standalone theorem.

\paragraph{Reduction to configurations.} Recall that in the first-moment argument in Section~\ref{sec:intro} we only considered lists whose differences are linearly independent, and that if the differences are dependent, the code contains them with a larger probability, but the number of such lists is smaller. To compare the two, we describe a list by a random vector with sparse coordinates under sparse linear constraints.

We say that an $\F_q$-valued random variable is \emph{$\rho$-sparse} if it is nonzero with probability at most $\rho$, and that a random vector is \emph{$\rho$-sparse} if each of its coordinates is $\rho$-sparse. A $\rho$-sparse random variable has entropy at most $h_q(\rho)$, and equality holds when it is nonzero with probability exactly $\rho$ and its nonzero values are equally likely. As in \cite{MRRSW}, we denote this distribution by $\mu_q(\rho)$.

Consider the matrix whose columns are the $m$ codewords of the list (Figure~\ref{fig:ones}). Each codeword is at relative distance at most $\rho$ from the center $y$ of the ball. Hence each column of the matrix obtained by subtracting $y$ from every column has at most $\rho n$ nonzero entries, and in a uniformly random row the entry of each such column is nonzero with probability at most $\rho$. The span of the rows of this matrix and of the all-ones vector has some dimension $k\le m$. The $L$ differences of the list span a subspace of dimension $k-1$, and $k = m$ exactly when they are linearly independent. Thus there are $k$ of the columns such that each of the other $N = m-k$ columns is a fixed linear combination of these $k$ columns. Let $X\in\F_q^k$ be the vector of the entries of these $k$ columns in a uniformly random row, and let $e_1,\dots,e_k$ denote the standard basis of $\F_q^k$. The entries of the other $N$ columns in this row are linear combinations $\ip{a_1}{X},\dots,\ip{a_N}{X}$ of the coordinates of $X$, and these combinations are also $\rho$-sparse.

We say that $a,b\in\F_q^k$ are \emph{proportional} if $a = cb$ for some $c\in\F_q\setminus\inset0$, and that $X$ and $a_1,\dots,a_N$ are a \emph{configuration} of dimension $k$ and length $k+N$ if the coordinates of $X$ and each $\ip{a_j}{X}$ are $\rho$-sparse, and the $k+N$ vectors $e_1,\dots,e_k,a_1,\dots,a_N$ are nonzero and pairwise non-proportional.

As in Section~\ref{sec:intro}, there are about $q^{nH_q(X)}$ choices of the codewords relative to the center, and a random linear code of rate $R$ contains the $k-1$ independent differences with probability $q^{-n(1-R)(k-1)}$. Thus we would expect the threshold rate for this list to be $1-H_q(X)/(k-1)$. In Lemma~\ref{lem:reduction} we show that indeed $\Rs\tau\ge1-H_q(X)/(k-1)$ for the type $\tau$ of the list, and that $X$ and $a_1,\dots,a_N$ are a configuration of dimension $k$ and length $m$.

To obtain $\Rs\tau\ge1-h_q(\rho)\inparen{1+\frac1L}$, the threshold rate that we would expect from Section~\ref{sec:intro}, we need $H_q(X)\le k\,h_q(\rho)-h_q(\rho)N/L$. When $N = 0$, as in Section~\ref{sec:intro}, $H_q(X)\le m\,h_q(\rho)$ is the trivial bound, since the entropy of each coordinate of $X$ is at most $h_q(\rho)$. When $N\ge1$, we need to tighten the trivial bound by $h_q(\rho)/L$ for each of the $N$ combinations, and so by at most $h_q(\rho)$, the entropy of one coordinate. Thus it suffices to show that $H_q(X)\le k\,h_q(\rho)-h_q(\rho)N/L$ for every configuration of dimension $k<m$ and length $m$.

\paragraph{Entropy under sparse linear constraints.}

\begin{restatable}[Entropy under sparse linear constraints]{theorem}{ToolTheorem}\label{thm:tool} Let $q$ be a prime power, $\rho\in\inparen{0,1-\frac1q}$, and let $X\in\F_q^k$ and $a_1,\dots,a_N\in\F_q^k$ be a configuration of dimension $k$ and length $k+N$. Let $\kappa := 1-\frac{q\rho}{q-1}$.
\begin{enumerate}
\item If $N\ge1$, then $H_q(X)\le k\,h_q(\rho)-\gamma^*$, where
$\gamma^* := \Dq{\rho}{\rho(1+\kappa)}>0$, and $D_q$ is the base-$q$ KL divergence between Bernoulli distributions.
\item For any $\eta>0$, there is $N_0 = N_0(q,\rho,\eta)$, polynomial in $\log(1/\eta)$, such that if $N\ge N_0$ then
\[
H_q(X)\;\le\;(k-1)\,h_q(\rho)+\eta.
\]
\end{enumerate}
\end{restatable}

Each $a_j$ has support of size at least two, since $a_j$ is nonzero and not proportional to any $e_i$. To see that the entropy of $X$ is smaller when $N\ge1$, note that if $X_1$ and $X_2$ are independent samples of $\mu_q(\rho)$, then $X_1+X_2$ is nonzero with probability $\rho(1+\kappa)>\rho$. If $X_1$, $X_2$ and $X_1+X_2$ are $\rho$-sparse, then $(X_1,X_2)$ does not have the distribution $\mu_q(\rho)^2$, and by part~(1) with $k = 2$ and $a_1 = e_1+e_2$ we have $H_q(X_1,X_2)\le2h_q(\rho)-\gamma^*$.

Informally, if there are many sparse combinations, then $X$ has the entropy of only $k-1$ coordinates, as if one coordinate of $X$ were determined by the others. As the following example shows, the $(k-1)\,h_q(\rho)$ term in part~(2) cannot be lowered. Let $X_1 = 0$, and let $X_2,\dots,X_k$ be independent samples of $\mu_q(\rho)$. Then the combinations $\ip{e_1+ce_j}{X} = cX_j$ for $c\in\F_q^\times$ and $2\le j\le k$ are $\rho$-sparse, and $H_q(X) = (k-1)\,h_q(\rho)$.

If $h_q(\rho)N\le\gamma^*L$, then $H_q(X)\le k\,h_q(\rho)-h_q(\rho)N/L$ by part~(1) of Theorem~\ref{thm:tool}, and otherwise by part~(2) for $L$ large enough, which establishes the lower bound in Theorem~\ref{thm:main-upper}. In other words, for large $L$, lists whose differences are dependent appear in the code only at rates at which lists whose differences are linearly independent appear as well. The proof of Theorem~\ref{thm:tool} is the most involved part of this work.

\FloatBarrier
\begin{figure}[!h]
\centering
\begin{tikzpicture}[x=0.85cm,y=0.85cm,>=latex,font=\scriptsize]
\draw[line width=0.4pt] (0,0) circle (1.05);
\fill (0,0) circle (1.3pt) node[below right=-1pt] {$y$};
\foreach \a/\r in {20/0.8,95/0.6,160/0.85,235/0.7,310/0.9}{\fill[black!60] (\a:\r) circle (1.6pt);}
\node[align=center] at (0,-1.5) {$m$ codewords in one ball\\ about the center $y$};
\draw[->,line width=0.4pt] (1.3,0) -- node[above] {subtract $y$} (2.9,0);
\begin{scope}[xshift=3.4cm,yshift=-1.2cm]
\foreach \c/\rows in {0/{1,5},1/{2,5},2/{0,3},3/{3},4/{1,2}}{
  \foreach \r in \rows {\fill[black!22] ({0.5*\c+0.12*(\c>2)},0.34*\r) rectangle ++(0.42,0.34);}}
\foreach \c in {0,...,4}{\foreach \r in {0,...,6}{\draw[line width=0.3pt] ({0.5*\c+0.12*(\c>2)},0.34*\r) rectangle ++(0.42,0.34);}}
\draw[decorate,decoration={brace,amplitude=3pt}] (0,2.45) -- node[above=3pt] {information set, $X\in\F_q^k$} (1.42,2.45);
\draw[decorate,decoration={brace,amplitude=3pt,mirror}] (1.62,-0.1) -- node[below=3pt] {$N = m-k$ combinations $\ip{a_j}{X}$} (2.54,-0.1);
\draw[decorate,decoration={brace,amplitude=3pt}] (2.67,2.38) -- node[right=3pt,align=left] {$n$\\coordinates} (2.67,0);
\end{scope}
\end{tikzpicture}
\caption{From $m$ codewords in one ball to a configuration (Lemma~\ref{lem:reduction}).}
\label{fig:ones}
\end{figure}
\FloatBarrier

\paragraph{The proof of Theorem~\ref{thm:tool}.} From the sparsity of the coordinates of $X$ we get $H_q(X)\le k\,h_q(\rho)$, and equality holds only when $X\sim\mu_q(\rho)^k$. To obtain a tighter bound, we also use the sparsity of the combinations. We first prove part~(1) using a single combination, and then part~(2) using many combinations.

\paragraph{Part~(1).} Under $\mu_q(\rho)^k$, as for $X_1+X_2$ after Theorem~\ref{thm:tool}, a combination of two or more coordinates is nonzero with probability at least $\rho(1+\kappa)>\rho$. If this combination is $\rho$-sparse, the probability that it is nonzero is at most $\rho$, and so $k\,h_q(\rho) - H_q(X)  \ge \Dq{\rho}{\rho(1+\kappa)} = \gamma^*$, which is part~(1). We make this precise using the Gibbs variational principle (Fact~\ref{fact:gibbs}), as in \cite{LM20}: $H_q(X)\le\E{\Phi(X)}+\log_qZ$ for any function $\Phi:\F_q^k\to\R_{\ge0}$, where $Z = \sum_xq^{-\Phi(x)}$. We take $\Phi(x)$ to be the sum of $\log_q\frac1z$ over the constraints $\ip{e_i}{x}$ and $\ip{a_j}{x}$ that are nonzero, where each constraint has its own $z\in(0,1]$. Since each constraint is $\rho$-sparse, its term in $\E{\Phi(X)}$ is at most $\rho\log_q\frac1z$, while the terms of $Z$ at which it is nonzero are multiplied by $z$. We choose $z$ for the $k$ coordinate constraints such that, with only these constraints, the bound is $k\,h_q(\rho)$ and the Gibbs measure $q^{-\Phi}/Z$ is $\mu_q(\rho)^k$ (Lemma~\ref{lem:calibration}). With one additional combination, the bound tightens by at least $\gamma^*$ for some $z<1$ (Corollary~\ref{cor:single-form}).

\paragraph{Part~(2).} Part~(2) follows if the bound on $H_q(X)$ tightens by $\gamma^*$ for each of about $h_q(\rho)/\gamma^*$ combinations, and it is easy to show this when the combinations are independent under $\mu_q(\rho)^k$. However, in general the combinations are not independent, and we show that either enough of them are close to independent, or, after a change of basis, many of them have supports of bounded size. In both cases $H_q(X)\le(k-1)\,h_q(\rho)+\eta$ for large enough $N$.

Under $\mu_q(\rho)^k$, a combination with support of size $w$ takes each value in $\F_q$ with probability $1/q$, up to $\kappa^w$ (Lemma~\ref{lem:activation}), and so it is close to uniform when $w$ is large. We say that the combinations are in \emph{general position} if no linear combination of two or more of them has small support (see Figure~\ref{fig:dichotomy}). In Section~\ref{sec:gp} we show that combinations in general position are close to independent. If more than $h_q(\rho)/\gamma^*$ of the combinations are in general position, then $H_q(X)<(k-1)\,h_q(\rho)$, as we show in Proposition~\ref{prop:gp} using an exact Fourier computation.

Otherwise, after a change of basis a large subset of the combinations has supports of bounded size (Section~\ref{sec:support-reduction}). We say that a collection of sets is a \emph{sunflower} if the intersection of any two of them is the same set, the \emph{core}. By the sunflower lemma of Erd\H{o}s and Rado \cite{ER60}, many of these supports form a sunflower (Proposition~\ref{prop:sunflower-bound} and Figure~\ref{fig:sunflower}). The part of each support that is not in the core, its \emph{petal}, is then disjoint from the other petals.

We now bound $H_q(X)$ by $(k-1)\,h_q(\rho)+\eta$, using the combinations with petals of size at least two if there are enough of them, and otherwise those with petals of size one. Since the petals are disjoint, the combinations are independent under $\mu_q(\rho)^k$ conditioned on the coordinates in the core. The bound $k\,h_q(\rho)$ tightens by at least $\gamma^*$ for each petal of size at least two (Proposition~\ref{prop:common-core}\ref{item:core-large}). All the combinations with petals of size one are nonzero at a coordinate $i^*$ of the core. Informally, since the petal coordinates are sparse, each of these combinations is usually equal to its core part, which is zero for exactly one value of $X_{i^*}$ given the other coordinates of the core. These combinations are also $\rho$-sparse. Hence their core parts are usually zero, and many such combinations nearly determine $X_{i^*}$. By Proposition~\ref{prop:common-core}\ref{item:core-single}, $H_q(X)$ is at most $(k-1)\,h_q(\rho)$, the bound for the coordinates other than $i^*$, up to an error term that is exponentially small in the number of these petals. We choose the size of the sunflower such that either at least $h_q(\rho)/\gamma^*$ of the petals have size at least two, or the number of petals of size one is large enough so that the error term is at most $\eta$. Thus $H_q(X)\le(k-1)\,h_q(\rho)+\eta$ for $N\ge N_0$, which establishes part~(2).

\FloatBarrier
\begin{figure}[!h]
\centering
\begin{tikzpicture}[x=0.85cm,y=0.85cm]
\begin{scope}
  \foreach \i in {0,3,5}{\fill[black!22] (0.3*\i,1.0) rectangle ++(0.3,0.3);}
  \foreach \i in {1,4,6}{\fill[black!22] (0.3*\i,0.6) rectangle ++(0.3,0.3);}
  \foreach \i in {0,1,3,4,5,6}{\fill[black!22] (0.3*\i,0.0) rectangle ++(0.3,0.3);}
  \foreach \y in {1.0,0.6,0.0}{\foreach \i in {0,...,7}{\draw[line width=0.3pt] (0.3*\i,\y) rectangle ++(0.3,0.3);}}
  \node[anchor=east,font=\scriptsize] at (-0.08,1.15) {$a_1$};
  \node[anchor=east,font=\scriptsize] at (-0.08,0.75) {$a_2$};
  \node[anchor=east,font=\scriptsize] at (-0.08,0.15) {$a_1+a_2$};
  \draw[line width=0.3pt] (0,0.45) -- (2.4,0.45);
  \node[font=\scriptsize] at (1.2,-0.35) {in general position};
\end{scope}
\begin{scope}[xshift=4.0cm]
  \foreach \i in {0,2,3,5,6}{\fill[black!22] (0.3*\i,1.0) rectangle ++(0.3,0.3);}
  \foreach \i in {0,2,3,5}{\fill[black!22] (0.3*\i,0.6) rectangle ++(0.3,0.3);}
  \foreach \i in {6}{\fill[black!22] (0.3*\i,0.0) rectangle ++(0.3,0.3);}
  \foreach \y in {1.0,0.6,0.0}{\foreach \i in {0,...,7}{\draw[line width=0.3pt] (0.3*\i,\y) rectangle ++(0.3,0.3);}}
  \node[anchor=east,font=\scriptsize] at (-0.08,1.15) {$a_1$};
  \node[anchor=east,font=\scriptsize] at (-0.08,0.75) {$a_2$};
  \node[anchor=east,font=\scriptsize] at (-0.08,0.15) {$a_1-a_2$};
  \draw[line width=0.3pt] (0,0.45) -- (2.4,0.45);
  \node[font=\scriptsize] at (1.2,-0.35) {a combination has small support};
\end{scope}
\begin{scope}[xshift=7.8cm]
  \foreach \i in {1,2}{\foreach \y in {1.0,0.6,0.2}{\fill[black!55] (0.3*\i,\y) rectangle ++(0.3,0.3);}}
  \foreach \i in {4,5}{\fill[black!22] (0.3*\i,1.0) rectangle ++(0.3,0.3);}
  \foreach \i in {6,7}{\fill[black!22] (0.3*\i,0.6) rectangle ++(0.3,0.3);}
  \foreach \i in {9}{\fill[black!22] (0.3*\i,0.2) rectangle ++(0.3,0.3);}
  \foreach \y in {1.0,0.6,0.2}{\foreach \i in {0,...,10}{\draw[line width=0.3pt] (0.3*\i,\y) rectangle ++(0.3,0.3);}}
  \node[anchor=east,font=\scriptsize] at (-0.08,1.15) {$a_1$};
  \node[anchor=east,font=\scriptsize] at (-0.08,0.75) {$a_2$};
  \node[anchor=east,font=\scriptsize] at (-0.08,0.35) {$a_3$};
  \draw[line width=0.3pt] (0.3,0.12) -- (0.9,0.12);
  \node[anchor=north,font=\scriptsize] at (0.6,0.1) {core};
  \draw[line width=0.3pt] (1.2,1.42) -- (3.0,1.42);
  \node[anchor=south,font=\scriptsize] at (2.1,1.44) {petals, pairwise disjoint};
\end{scope}
\end{tikzpicture}
\caption{Left and middle: the two cases of Lemma~\ref{lem:dichotomy}, with the support of
a combination shaded. Right: a sunflower among the supports (Fact~\ref{fact:sunflower}), with one
common core, shaded dark, and pairwise disjoint petals, the last of size one.}
\label{fig:dichotomy}\label{fig:sunflower}
\end{figure}
\FloatBarrier

\paragraph{Vectors in a Hamming ball.}
Let $B\subseteq\F_q^n$ denote the Hamming ball of relative radius $\rho$ centered at $0$. For $q=2$, Guruswami, H\r{a}stad and Kopparty \cite[Lemma~7]{GHK11} showed that two uniformly random vectors in $B$ sum to a vector in $B$ with probability at most $2^{-\Omega(n)}$. Partitioning the sequences in $B^k$ according to the type of the matrix with columns $x_1,\dots,x_k$, we get the following generalization of their lemma, which holds for any $q$, for any number of vectors, and for many linear combinations of them at the same time. We prove Corollary~\ref{cor:hamming} in Appendix~\ref{app:deferred}.

\begin{corollary}\label{cor:hamming}
Let $a_1,\dots,a_N\in\F_q^k$ be such that $e_1,\dots,e_k,a_1,\dots,a_N$ are nonzero and pairwise non-proportional, and let $x_1,\dots,x_k$ be independent and uniform in $B$. Let $E$ be the event that $\sum_ia_{j,i}x_i\in B$ for all $j\le N$, where $a_{j,i}$ is the $i$-th coordinate of $a_j$. Then $\PR E\le q^{-\gamma^*n+o(n)}$ if $N\ge1$, and, for any $\eta>0$, $\PR E\le q^{-(h_q(\rho)-\eta)n+o(n)}$ if $N\ge N_0(q,\rho,\eta)$.
\end{corollary}

\subsubsection{Random LDPC codes}\label{sec:overview-ldpc}

Let $\tau$ be a scale-invariant type with full support, and let $X$ be the number of matrices of type $\tau_n$ in a random $s$-LDPC code of rate $R$. To prove Theorem~\ref{thm:ldpc-core}, we show that $X>0$ with probability $1-o(1)$. The expectation $\E X$ is close to the expected number of matrices of type $\tau_n$ in a random linear code of the same rate (Lemma~\ref{lem:ldpc-one}), which is large since $R>\Rs\tau$. However, as in Section~\ref{sec:intro}, a large expectation does not mean that $X>0$ with high probability. We use the second moment method. By the Cauchy--Schwarz inequality, $\PR{X>0}\ge(\E X)^2/\E{X^2}$, and so it suffices to show that $\E{X^2}\le(1+o(1))(\E X)^2$. For random LDPC codes this bound does not hold. We first show that $X>0$ with constant probability. Next, we show that, apart from the codes with a nonzero low-weight codeword, the second moment is larger than the square of the expectation only because $X$ depends on the numbers of short cycles of the \emph{Tanner graph}, the bipartite graph between the coordinates and the parity checks, which implies that $X>0$ with probability $1-o(1)$.

\paragraph{Constant probability.} The second moment $\E{X^2}$ is the sum, over pairs of matrices of type $\tau_n$, of the probability that the code contains both matrices of the pair. For most pairs, the joint type is close to $\tau\otimes\tau$, the joint type of two independent matrices of type $\tau$, and for such a pair we show that this probability is roughly the product of the probabilities for the two matrices. However, a code with a nonzero low-weight codeword $x$ contains many pairs of matrices in which one matrix is obtained by adding $x$ to some of the columns of the other. The contribution of these codes to $\E{X^2}$ is too large. Since a random LDPC code is unlikely to have such a codeword, we let $Y := X$ if the code has no such codeword, and $Y := 0$ otherwise. By Lemma~\ref{lem:ldpc-planted}, such a codeword is unlikely even conditioned on the code containing a given matrix of type $\tau_n$. Hence $\E Y = (1-o(1))\E X$.

For random linear codes, the probability that the code contains a given pair of matrices depends only on the rank of the pair, but for random LDPC codes it depends on the joint type. We need to compute this probability up to a $1+o(1)$ factor for most pairs. To bound the probability for one matrix, \cite{MRRSW} view the $s$ rows of a block as independent samples of $\tau$. This approximation is correct only up to a factor that is exponential in $n$. Thus if we use this approximation for pairs, we only get $\E{X^2}\le q^{O(n)}(\E X)^2$.

To compute the probability for one matrix, let $\Lambda$ be a random matrix whose rows are independent samples of a distribution on $\F_q^m$, and let $\mathscr P$ and $\mathscr T$ be the events that $\Lambda$ satisfies the parity checks of one layer and that $\Lambda$ has type $\tau_n$. The probability that $\Lambda$ equals a fixed matrix depends only on the type of the matrix, as does the probability that the matrix satisfies the checks. Hence for any row distribution the probability for one matrix is $\PR{\mathscr P\mid\mathscr T} = \PR{\mathscr P}\,\PR{\mathscr T\mid\mathscr P}/\PR{\mathscr T}$ (Lemma~\ref{lem:ldpc-layer}), where $\PR{\mathscr P}$ and $\PR{\mathscr T}$ are easy to compute. Given $\mathscr P$, $n$ times the type of $\Lambda$ is a sum of $n/s$ independent random vectors, since the blocks of $\Lambda$ are independent. If the rows of $\Lambda$ are samples of $\tau$, then the rows of a block that satisfies its parity check do not in general have distribution $\tau$, and $\PR{\mathscr T\mid\mathscr P}$ is exponentially small in $n$. Instead, let the rows of $\Lambda$ be samples of a distribution $\nu$ close to $\tau$ such that the rows of such a block have distribution $\tau$ (Lemma~\ref{lem:ldpc-smooth}). Then, given $\mathscr P$, the expectation of $n$ times the type of $\Lambda$ is $n\tau$, and we compute $\PR{\mathscr T\mid\mathscr P}$ up to a $1+o(1)$ factor using a local limit theorem (Fact~\ref{fact:llt}).

For a pair whose joint type is close to $\tau\otimes\tau$, we replace $\nu$ by $\nu\otimes\nu$. With fixed coefficients, a block whose rows are independent samples of $\nu\otimes\nu$ satisfies its parity check in the first $m$ and the last $m$ columns independently. Thus the probability that the code contains such a pair is roughly the square of the probability for one matrix (Lemma~\ref{lem:ldpc-product}). This probability differs from the square because, given its parity check, the rows of such a block are dependent. In Proposition~\ref{prop:ldpc-second} we show that $\E{Y^2}\le(C_s+o(1))(\E Y)^2$. Hence $X\ge Y>0$ with probability at least $1/C_s-o(1)$.

\paragraph{Probability $1-o(1)$.} For a fixed $s$ we only get that $X>0$ with constant probability. Indeed, $\E{Y^2}/(\E Y)^2$ does not tend to $1$, since conditioned on the code containing a given matrix $M$ of type $\tau_n$, the numbers of short cycles of the Tanner graph have means different from those in a random $s$-LDPC code. For $q\ge3$ these means also depend on the \emph{coefficient product} of a cycle, the product of the coefficient ratios around it. For example, a cycle of length $4$ consists of two coordinates $u$ and $v$ that both belong to the same two checks, and if the ratio of the coefficients of $v$ and $u$ is $r$ in one check and $r'$ in the other, then the coefficient product is $r/r'$. Conditioned on $M$ satisfying a check in which this ratio is $r$, the pair $(a,rb)$, where $a$ and $b$ are the rows of $M$ at $u$ and $v$, has the same distribution as the pair of rows at two coordinates with equal coefficients. Hence $a$ depends on $rb$ in one check and on $r'b$ in the other, and $rb = r'b$ for every $b$ exactly when the coefficient product is $1$. Thus, conditioned on $M$, the mean number of cycles of length $4$ depends on their coefficient product, and in Lemma~\ref{lem:ldpc-cycles} we count the cycles of each length and coefficient product separately.

To show that $X>0$ with probability $1-o(1)$, we use the small subgraph conditioning theorem of Janson \cite{Jan95}. Informally, this theorem states that if $\E{Y^2}$ is larger than $(\E Y)^2$ only because $Y$ depends on the numbers of short cycles, then $Y>0$ with probability $1-o(1)$. More precisely, if conditioning on $M$ changes the mean number of cycles of a given length and coefficient product from $\varpi$ to $\varpi(1+\beta)$, then it suffices that $\E{Y^2}\le\inparen{\exp{\sum\varpi\beta^2}+o(1)}(\E Y)^2$, where the sum is over the lengths and the coefficient products. We show that $C_s = \exp{\sum\varpi\beta^2}$ (Proposition~\ref{prop:ldpc-identity}). Hence, by Proposition~\ref{prop:ldpc-second}, $X\ge Y>0$ with probability $1-o(1)$. This identity is the most interesting part of the proof, at least to the author, and it fails in general if we ignore the coefficient products of the cycles.

\subsection{Discussion and open problems}\label{sec:discussion}
We prove Theorem~\ref{thm:main-twopoint} for small $\eps$, but we believe that the two-point concentration holds for a large range of $\eps$. To see why our proof does not extend to larger $\eps$, note that in the proof of Theorem~\ref{thm:main-upper} we obtain a configuration $X$, $a_1,\dots,a_N$ from a list of type $\tau$. Every combination of this configuration has coefficient sum $1$, since the all-ones vector satisfies the same linear relations as the rows (see the proof of Lemma~\ref{lem:reduction}). By Lemma~\ref{lem:reduction}, $\Rs\tau\ge1-H_q(X\bmod\F_q\onesvector)/(k-1)\ge1-H_q(X)/(k-1)$. We only bound $H_q(X)$. As the following example shows, the second inequality is not tight.

Let $\kappa$ be as in Theorem~\ref{thm:tool}, $2\le k\le\frac1\kappa$, and let $X$ be the random vector in $\F_q^k$ which is $0$ with probability $\kappa$ and uniform on $\F_q^k\setminus\F_q\onesvector$ with probability $1-\kappa$. Every $a\in\F_q^k$ with coefficient sum $1$ has $\PR{\ip aX\neq0} = \rho$, and any two of these vectors are linearly independent, so $X$ and the vectors with coefficient sum $1$ other than $e_1,\dots,e_k$ are a configuration. There are $q^{k-1}$ vectors with coefficient sum $1$, and so this configuration has length $q^{k-1}$. Note that $H_q(X)>k-1$. Hence $1-H_q(X)/(k-1)$ is negative, although $\Rs\tau\ge0$ for any type $\tau$. Thus for lists of size $q^{k-1}$ with $k$ close to $\frac1\kappa$ we get no lower bound on $\Rs\tau$ by bounding $H_q(X)$, and since $\eps_0 = h_q(\rho)/L_0$, we cannot get $\eps_0$ larger than $e^{-\Omega_q(\kappa^{-1})}$. We conjecture the following bound on $H_q\inparen{X\bmod\F_q\onesvector}$.

\begin{conjecture}\label{conj:range}
There is $C_q$ such that the following holds for any
$\rho\in\inparen{0,1-\frac1q}$. Let $X$ and $a_1,\dots,a_N$ be a configuration of dimension $k$ and length $m$ in which every combination has coefficient sum $1$, and suppose that $m\inparen{1-h_q(\rho)}\ge C_q$. Then
\[
H_q\inparen{X\bmod\F_q\onesvector}\;\le\;(k-1)\,h_q(\rho)\,\frac{m}{m-1}.
\]
\end{conjecture}

The conjecture would imply that $\TRLC(\rho,L)\ge1-h_q(\rho)\inparen{1+\frac1L}$ whenever $(L+1)(1-h_q(\rho))\ge C_q$, since every combination of the configuration obtained from a list has coefficient sum $1$.

\section{Preliminaries}\label{sec:prelims}
Fix a prime power $q$ and a radius
$\rho\in\inparen{0,1-\frac1q}$.

\subsection{Entropy and divergence}\label{sec:entropies}

For a distribution $\mu$, $H_q(\mu)$ denotes the entropy of a random variable distributed according to $\mu$. The \emph{$q$-ary entropy function}
\[
h_q(a) := a\log_q(q-1) - a\log_q a - (1-a)\log_q(1-a)
\qquad (0\le a\le1)
\]
is the entropy of a random variable on $\F_q$ which is $0$ with probability $1-a$ and uniform on $\F_q\setminus\inset0$ with probability $a$. Finally,
\[
\Dq{a}{b} := \frac{1}{\ln q}\inbrak{a\ln\frac ab + (1-a)\ln\frac{1-a}{1-b}}
\qquad (0<a,b<1)
\]
is the base-$q$ Kullback--Leibler divergence
$D\inparen{\Ber(a)\,\|\,\Ber(b)}$.

Denote $h := h_q(\rho)$. The function $h_q$ is strictly increasing in $\inbrak{0,1-\frac1q}$, and $h_q\inparen{1-\frac1q} = 1$. The divergence is at least $0$ and vanishes only at $a = b$, and for fixed $a$ it is strictly increasing in $b$ on $[a,1)$.

\subsection{Sparsity}\label{sec:sparsity}

We use $\rho$-sparse random variables and vectors, and the distribution $\mu_q(\rho)$, as defined in Section~\ref{sec:overview-rlc}. As in \cite{MRRSW}, we let $\mu_q = \mu_q(\rho)$, and we write $\mu_q^d$ for the distribution of a random vector in $\F_q^d$ whose entries are independent samples of $\mu_q$. Recall also that $\kappa = 1-\frac{q\rho}{q-1}$, and note that $\rho<1-\frac1q$ implies $\kappa\in(0,1)$.

To compute the probability that a linear combination of the coordinates of a random vector vanishes, we use Fourier analysis as in \cite[Section~4.1]{MRRSW}. Let $p$ be the prime such that $q$ is a power of $p$, let $\mathrm{Tr}:\F_q\to\F_p$ be the trace map, and for $x,y\in\F_q^d$ let $\chi_x(y) := \omega_p^{\mathrm{Tr}(\ip xy)}$, where $\omega_p = e^{2\pi i/p}$. For a function $f:\F_q^d\to\CC$, we define the Fourier transform $\hat f:\F_q^d\to\CC$ of $f$ by $\hat f(y) := \Eover{x\in\F_q^d}{f(x)\,\overline{\chi_x(y)}}$. Then we have $f(x) = \sum_{y\in\F_q^d}\hat f(y)\,\chi_y(x)$. In particular, $\one\inbrak{y = 0} = \frac1q\sum_{c\in\F_q}\chi_c(y)$ for $y\in\F_q$.

\subsection{Types}\label{sec:types-def}

\begin{definition}[Type]\label{def:type}
Let $m\ge1$. A \emph{type of length $m$} is a distribution $\tau$ on $\F_q^m$, and a matrix $B\in\F_q^{n\times m}$ is \emph{of type $\tau$} if its row distribution is $\tau$. We let $x\sim\tau$, $C = C(\tau):= \spn(\supp\tau)\le\F_q^m$ and $\dim(\tau):= \dim C$. The type $\tau$ is \emph{$\rho$-clustered} if $\PR{x_i\neq x_j}>0$ for $i\neq j$ and there is an $\F_q$-valued random variable $\alpha$ (a \emph{center}), defined on the same probability space as $x$ and not necessarily a function of $x$, such that
\[
\PR{x_i\neq \alpha}\;\le\;\rho\qquad\text{for any } i\in[m].
\]
\end{definition}

\subsection{Implied types and the threshold rate}

For ease of exposition, in the proofs we use an equivalent form of the implied types: for a type $\tau$ and a subspace $V\le C$, the \emph{implied type} $\tau\bmod V$ is the distribution of $x+V$ in $C/V$, which has dimension $\dim C-\dim V$. The implied types $\tau\bmod V$ with $V\lneq C$ are, up to isomorphism, the $\tau$-implied types of nonzero dimension \cite[Definition~2.6]{MRRSW}.

\begin{definition}\label{def:rat}
For a type $\tau$ with $\dim(\tau) = r\ge1$ let
\[
\Rs\tau:= \max_{V\lneq C}\inparen{1-\frac{H_q(\tau\bmod V)}{r-\dim V}}
\;\in\;[0,1].
\]
Let $\TRLC(\rho,L):= \inf\Rs\tau$, where the infimum is over all $\rho$-clustered types $\tau$ of length $m = L+1$.
\end{definition}

By \cite[Lemma~2.7]{MRRSW}, for any $\gamma>0$ and any type $\tau$ whose probabilities are in $\frac1n\Z$, a random linear code of rate $R$ contains a matrix of type $\tau$ with probability at most $q^{-\gamma n}$ if $R\le\Rs\tau-\gamma$, and with probability at least $1-\poly(n)\,q^{-\gamma n}$ if $R\ge\Rs\tau+\gamma$. In \cite{MRRSW}, a random linear code is the kernel of a uniformly random parity-check matrix instead of the image of a uniformly random generator matrix, and a property that holds with probability $1-o(1)$ for one of them holds with probability $1-o(1)$ for the other \cite[Footnote~7]{MRRSW}.

\begin{lemma}[Threshold for list-decodability]\label{fact:threshold}
Fix $L\ge1$, let $m:= L+1$, and let
$\cC\le\F_q^n$ be a random linear code of rate $R$.
\begin{enumerate}[label=\textup{(\roman*)}]
\item\label{item:threshold-sub} If $R<\TRLC(\rho,L)$, then with probability $1-o(1)$ as $n\to\infty$, $\cC$ is
$(\rho,L)$-list-decodable.
\item\label{item:threshold-super} If $\tau$ is a $\rho$-clustered type of length $m$ and
$R>\Rs\tau$, then with probability $1-o(1)$ as $n\to\infty$, $\cC$ is not $(\rho,L)$-list-decodable.
\end{enumerate}
\end{lemma}

\begin{proof}
We prove the lemma for the random linear code of \cite{MRRSW}. By \cite[proof of Theorem~1.2]{MRRSW}, $\cC$ is not $(\rho,L)$-list-decodable if and only if $\cC$ contains a matrix whose $m$ distinct columns are contained in some Hamming ball of relative radius $\rho$. The types of these matrices are $\rho$-clustered types of length $m$ with probabilities in $\frac1n\Z$. Each of them has $\Rs\tau\ge\TRLC(\rho,L)$. Part~\ref{item:threshold-sub} then follows from \cite[Lemma~2.7]{MRRSW} with $\gamma:= \TRLC(\rho,L)-R$ and a union bound over the $\poly(n)$ types. For part~\ref{item:threshold-super}, let $\tau_n$ be the distribution of $x$ after the probabilities of $(x,\alpha)$ are rounded down to multiples of $\frac1n$ and the remaining weight is put on $(x,\alpha) = (0,0)$. The probabilities $\PR{x_i\neq\alpha}$ do not increase, and for $n$ large enough the matrices of type $\tau_n$ still have distinct columns, so $\tau_n$ is the type of such a matrix. For $n$ large we have $\supp\tau\subseteq\supp\tau_n\subseteq\supp\tau\cup\inset0$, so $C(\tau_n) = C(\tau)$. Note that $\tau\mapsto\Rs\tau$ is the maximum of finitely many functions $\tau\mapsto1-H_q(\tau\bmod V)/(r-\dim V)$, each continuous in the probabilities. Hence $\Rs{\tau_n}\to\Rs\tau$. For some fixed $\gamma>0$, $R\ge\Rs{\tau_n}+\gamma$ for all large enough $n$, and part~\ref{item:threshold-super} follows from \cite[Lemma~2.7]{MRRSW}, as in \cite{GLMRSW}.
\end{proof}

\section{Proofs of \texorpdfstring{Theorems~\ref{thm:main-twopoint} and~\ref{thm:main-upper}}{Theorems 1.1 and 1.3}}\label{sec:highlevel}
In this section we prove Theorems~\ref{thm:main-twopoint} and~\ref{thm:main-upper} based on Theorem~\ref{thm:tool} and Proposition~\ref{prop:sigma-min}.

\subsection{The upper bound of \texorpdfstring{Theorem~\ref{thm:main-upper}}{Theorem 1.3}}\label{sec:hl-lower}

\begin{definition}[{\cite[Definition~4.2]{GLMRSW}}]\label{def:sigma}
Let $m\ge2$, let $\onesvector:= (1,\dots,1)\in\F_q^m$, let $u\sim\mu_q^m$, and let $\alpha$ be uniform on $\F_q$, independent of $u$. Let $\sigma_m$ be the distribution of
\[
Y:= u+\alpha\onesvector\ \in\F_q^m.
\]
\end{definition}

Note that $\sigma_m$ is a $\rho$-clustered type of full support with center $\alpha$. Indeed, $\PR{Y_i\neq\alpha} = \PR{u_i\neq0} = \rho$ and $\PR{Y_i\neq Y_j}>0$ for $i\neq j$.

\begin{lemma}\label{lem:delta}
Let $\Delta_m:= H_q(u)-H_q\inparen{u\bmod\F_q\onesvector}$. Then
\[
0\;\le\;\Delta_m\;\le\;\frac{2(q-1)}{\ln q}\,\lambda^m,
\]
where $\lambda := 2\sqrt{\frac{\rho(1-\rho)}{q-1}}+\frac{(q-2)\rho}{q-1}\in(0,1)$.
\end{lemma}

We prove Lemma~\ref{lem:delta} in Appendix~\ref{app:deferred}.

\begin{restatable}{proposition}{SigmaMin}\label{prop:sigma-min} Let $m\ge2$. Then
\[
\Rs{\sigma_m}
\;=\;1-h_q(\rho)\inparen{1+\frac1{m-1}}+\frac{\Delta_m}{m-1},
\]
and the maximum in Definition~\ref{def:rat} is attained at
$V = \F_q\onesvector$.
\end{restatable}

We prove Proposition~\ref{prop:sigma-min} in Section~\ref{sec:lower}.

\subsection{Proof of \texorpdfstring{Theorem~\ref{thm:main-upper}}{Theorem 1.3}}\label{sec:hl-upper}\label{sec:vabs-upper}\label{sec:assembly}

Throughout this subsection $\tau$ is a
$\rho$-clustered type of length $m = L+1$ and dimension $r$ with center
$\alpha$.

\begin{restatable}{lemma}{ReductionLemma}\label{lem:reduction} Let $\tau$ be a $\rho$-clustered type of length $m$, and let $C:=\spn(\supp\tau)$ and $k:= \dim(C+\F_q\onesvector)$. Then there is a configuration $X$, $a_1,\dots,a_{m-k}$ of dimension $k$ and length $m$ (see Section~\ref{sec:overview-rlc}) with
\[
\Rs\tau\;\ge\;1-\frac{H_q\inparen{X\bmod\F_q\onesvector}}{k-1}
\;\ge\;1-\frac{H_q(X)}{k-1}.
\]
\end{restatable}

\begin{proof}[Proof of Lemma~\ref{lem:reduction}]
Let $C':= C+\F_q\onesvector$, which has dimension $k$, and let $\xi_i\in(C')^*$ be the restriction to $C'$ of the $i$-th coordinate map $\F_q^m\to\F_q$. Each $\xi_i$ takes the value $1$ at $\onesvector$, and so $\xi_i = c\,\xi_j$ implies $c = 1$ and $x_i = x_j$ almost surely, which is impossible because $\PR{x_i\neq x_j}>0$. Thus the $\xi_i$ are nonzero and pairwise non-proportional. Note that the $\xi_i$ span $(C')^*$, because any linear map $C'\to\F_q$ extends to $\F_q^m$.

Let $\bar x:= x-\alpha\onesvector\in C'$. Each coordinate $\bar x_i = x_i-\alpha$ is $\rho$-sparse, because $\PR{x_i\neq\alpha}\le\rho$. Choose an \emph{information set} $I\subseteq[m]$ of size $k$, that is, a set such that $\inset{\xi_i}_{i\in I}$ is a basis for $(C')^*$. Then the projection $C'\to\F_q^I$ is injective. Let $X:= \bar x_I\in\F_q^I\cong\F_q^k$. Note that $\bar x$ is a function of $X$, and that the $k$ coordinates of $X$ are $\rho$-sparse. For $j\notin I$ we can write $\xi_j = \sum_{i\in I}a_{ji}\,\xi_i$. Then $\bar x_j = \ip{a_j}{X}$, where $a_j:= (a_{ji})_{i\in I}$. In particular, $\ip{a_j}{X}$ is $\rho$-sparse. Since each $\xi_i$ takes the value $1$ at $\onesvector$, each $a_j$ has coefficient sum $1$. Taking coordinates in the basis $\inset{\xi_i}_{i\in I}$ gives an isomorphism $(C')^*\to\F_q^k$, which maps $\xi_i$ to $e_i$ for $i\in I$ and $\xi_j$ to $a_j$ for $j\notin I$. Thus the $m$ vectors $e_i$ and $a_j$ are nonzero and pairwise non-proportional, and we have a configuration of dimension $k$ and length $m$.

It remains to bound $\Rs\tau$. Let $V:= C\cap\F_q\onesvector$. Then $V\neq C$, because $C$ contains a vector that is not in $\F_q\onesvector$, and $r-\dim V = \dim C'-1 = k-1$. For $x,x'\in C$ we have $x-x'\in\F_q\onesvector$ if and only if $x-x'\in V$. Hence each of $x\bmod V$ and $x\bmod\F_q\onesvector$ is a function of the other. Note that the image of $\onesvector$ under the projection $C'\to\F_q^I$ is $\onesvector$, and so each of $X\bmod\F_q\onesvector$ and $\bar x\bmod\F_q\onesvector = x\bmod\F_q\onesvector$ is a function of the other. Taking this $V$ in the maximum of Definition~\ref{def:rat} gives $\Rs\tau\ge1-H_q(\tau\bmod V)/(k-1) = 1-H_q(X\bmod\F_q\onesvector)/(k-1)$. The second inequality holds because $X\bmod\F_q\onesvector$ is a function of $X$.
\end{proof}

In the proof of Theorem~\ref{thm:main-upper} we will apply part~(2) of Theorem~\ref{thm:tool} with $\eta = h/L$. Let $L_0 = L_0(q,\rho)$ be the minimum integer such that $L_0\ge2$ and
\begin{equation}\label{eq:choice-L0prime}
\frac{\gamma^*L}{h}\;\ge\;N_0\inparen{q,\rho,\frac hL}
\qquad\text{and}\qquad
\Delta_L\;<\;\frac h{2L}
\qquad\text{for all }L\ge L_0,
\end{equation}
where $N_0$ is as in Theorem~\ref{thm:tool}. Note that $L_0$ exists. Indeed, by Theorem~\ref{thm:tool}, $N_0(q,\rho,h/L)$ is polynomial in $\log L$, while the left-hand side of the first inequality grows linearly in $L$, and $L\Delta_L$ tends to $0$ by Lemma~\ref{lem:delta}.

\begin{proof}[Proof of Theorem~\ref{thm:main-upper}]
Let $L\ge L_0$. For the upper bound, note that by Proposition~\ref{prop:sigma-min}, $\TRLC(\rho,L)\le\Rs{\sigma_{L+1}} = 1-h_q(\rho)\inparen{1+\frac1L}+\frac{\Delta_{L+1}}{L}$, and $\Delta_{L+1} = q^{-\Omega_{q,\rho}(L)}$ by Lemma~\ref{lem:delta}. For the lower bound, we show that $\Rs\tau\ge1-h\inparen{1+\frac1L}$.

By Lemma~\ref{lem:reduction}, there is a configuration of dimension $k$ with $N = m-k$ combinations. Note that $C$ contains a vector that is not constant, since $\PR{x_i\neq x_j}>0$ for $i\neq j$. Hence $2\le k\le m$ and $0\le N\le L-1$. By the same lemma, it suffices to show that $H_q(X)\le(k-1)\,h\inparen{1+\frac1L}$, that is,
\[
H_q(X)\;\le\;k\,h-\frac{hN}{L}.
\]
If $N = 0$, then $H_q(X)\le\sum_iH_q(X_i)\le k\,h$, since the coordinates of $X$ are $\rho$-sparse. If $1\le N$ and $hN\le\gamma^*L$, then by part~(1) of Theorem~\ref{thm:tool} we have $H_q(X)\le k\,h-\gamma^*\le k\,h-\frac{hN}{L}$. Otherwise $N>\gamma^*L/h\ge N_0(q,\rho,h/L)$ by \eqref{eq:choice-L0prime}, and by part~(2) of Theorem~\ref{thm:tool} with $\eta = h/L$ we have
\[
H_q(X)\;\le\;(k-1)\,h+\frac hL\;=\;k\,h-h\inparen{1-\frac1L}\;\le\;k\,h-\frac{hN}{L},
\]
where we used $N\le L-1$ in the last inequality.
\end{proof}

\subsection{Proof of \texorpdfstring{Theorem~\ref{thm:main-twopoint}}{Theorem 1.1}}\label{sec:derive}

In Theorem~\ref{thm:main-twopoint} we take $\eps_0:= h_q(\rho)/L_0$, with $L_0$ as in \eqref{eq:choice-L0prime}, and $c>0$ such that $\zeta\Delta_{\floor\zeta+1}/h_q(\rho)\le q^{-c/\eps}$ for all $0<\eps<\eps_0$, which exists by Lemma~\ref{lem:delta} and \eqref{eq:choice-L0prime}.

\begin{proof}[Proof of Theorem~\ref{thm:main-twopoint} and Corollary~\ref{cor:all-eps}]
We first prove Corollary~\ref{cor:all-eps}. Let $L:= \max\inset{\floor\zeta+1,\,L_0}$. Then $L\ge L_0$ and $L>\zeta$. By Theorem~\ref{thm:main-upper} we have $\TRLC(\rho,L)\ge1-h\inparen{1+\frac1L}$, and $h/L<\eps$. Hence $R = 1-h-\eps<\TRLC(\rho,L)$, and $\cC$ is $(\rho,L)$-list-decodable by Lemma~\ref{fact:threshold}\ref{item:threshold-sub}. Note that $L\le\zeta+L_0(q,\rho)$, so we may take the $O_{q,\rho}(1)$ term in Corollary~\ref{cor:all-eps} to be $L_0(q,\rho)$.

Next, let $0<\eps<\eps_0$. Then $\zeta>L_0$, and so $L = \floor\zeta+1$. Thus $L^*\le\floor\zeta+1$.

Finally, we prove the lower bounds. Let $m\in\inset{\floor\zeta,\floor\zeta+1}$. Then $m\ge L_0\ge2$. By Proposition~\ref{prop:sigma-min},
\[
R-\Rs{\sigma_m}\;=\;\frac{h}{m-1}-\frac{\Delta_m}{m-1}-\frac h\zeta\;=\;\frac1{m-1}\inparen{\frac{h\,(\zeta-m+1)}{\zeta}-\Delta_m}.
\]
For $m = \floor\zeta$ we have $\zeta-m+1\ge1$ and $\zeta<2m$. By \eqref{eq:choice-L0prime} we have $\Delta_m<h/(2m)<h/\zeta$. Hence $R>\Rs{\sigma_m}$. For $m = \floor\zeta+1$ we have $\zeta-m+1 = \fracpart\zeta$, and so $R>\Rs{\sigma_m}$ when $\fracpart\zeta>q^{-c/\eps}\ge\zeta\Delta_m/h$. Whenever $R>\Rs{\sigma_m}$, Lemma~\ref{fact:threshold}\ref{item:threshold-super} implies that $\cC$ is not $(\rho,m-1)$-list-decodable, that is, $L^*\ge m$. In particular, $L^*\ge\floor\zeta$, and $L^*\ge\floor\zeta+1$ when $\fracpart\zeta>q^{-c/\eps}$. Since also $L^*\le\floor\zeta+1$, we have $L^* = \floor\zeta+1 = \ceil\zeta$ when $\fracpart\zeta>q^{-c/\eps}$.
\end{proof}

\section{Proof of \texorpdfstring{Theorem~\ref{thm:tool}}{Theorem 1.6}}\label{sec:upper}

In this section we prove Theorem~\ref{thm:tool}, following the outline in Section~\ref{sec:overview}. Throughout this section $X$ is a random vector in $\F_q^k$ all of whose
$k$ coordinates are $\rho$-sparse. In Section~\ref{sec:tilting} we bound $H_q(X)$ using the Gibbs variational principle, and we prove part~(1) using a single combination. For part~(2), in Section~\ref{sec:gp} we bound $H_q(X)$ when enough of the combinations are in general position. Otherwise, after a change of basis many of the combinations have small support (Section~\ref{sec:support-reduction}), and we find a sunflower among their supports (Section~\ref{sec:endgame}) and apply the bounds of Section~\ref{sec:common-core} to the corresponding combinations. We complete the proof in Section~\ref{sec:low}.

\subsection{The Gibbs variational principle and a single combination}\label{sec:tilting}
We bound the entropy using the Gibbs variational principle (Fact~\ref{fact:gibbs}) as in \cite{LM20}. We choose $\Phi$ so that the bound obtained by using only the $k$ coordinate constraints is $k\,h_q(\rho)$, and so that equality holds for $\mu_q^k$ (Section~\ref{sec:sparsity}). The proofs of Lemmas~\ref{lem:activation} and~\ref{lem:tilt-opt} consist of straightforward computations, and we defer them to Appendix~\ref{app:deferred}.

\begin{fact}[Gibbs variational principle]\label{fact:gibbs} Let $\mathcal{X}$ be a finite set, let $\mu$ be a probability measure on $\mathcal{X}$, and let $\Phi:\mathcal{X}\to\R_{\ge0}$ be a function. Then
\[
H_q(\mu)\;\le\;\Eover{\mu}{\Phi} + \log_q Z, \qquad
Z:= \sum_{x\in\mathcal{X}}q^{-\Phi(x)},
\]
and equality holds if and only if $\mu$ is the Gibbs measure $q^{-\Phi}/Z$.
\end{fact}

Fact~\ref{fact:gibbs} is well known, and it follows from $\sum_x\mu(x)\log_q\frac{\mu(x)}{\nu(x)}\ge0$ for $\nu(x):= q^{-\Phi(x)}/Z$.

We use Fact~\ref{fact:gibbs} only when $\mathcal{X} = \F_q^k$, $\mu$ is the distribution of $X$, and
\[
\Phi(x) = \sum_{p=1}^{P}\log_q\frac1{z_p}\,\one\inbrak{\ip{a^{(p)}}{x}\neq0},
\qquad z_p\in(0,1],\ a^{(p)}\in\F_q^k,
\]
where each $\ip{a^{(p)}}{X}$ is $\rho$-sparse. Hence $\Eover\mu\Phi = \sum_p\log_q\frac1{z_p}\PR{\ip{a^{(p)}}{X}\neq0}\le \rho\sum_p\log_q\frac1{z_p}$ and
\begin{equation}\label{eq:tilt}
H_q(X)\;\le\;\rho\sum_{p=1}^P\log_q\frac1{z_p}\;+\;\log_q Z,
\qquad
Z = \sum_{x\in\F_q^k}
\prod_{p=1}^Pz_p^{\one\inbrak{\ip{a^{(p)}}{x}\neq0}}.
\end{equation}
Our use of Fact~\ref{fact:gibbs} is similar to the proof in \cite[Section~5.2]{MRRSW}, where the probability that a random LDPC code contains a given vector of a given weight is bounded using a random vector whose entries are independent samples of $\mu_q(\beta)$, and $\beta$ is then chosen to minimize the bound.

\begin{lemma}\label{lem:calibration}
Let $z_0:= \frac{\rho}{(q-1)(1-\rho)}$. Then $z_0<1$, $1+(q-1)z_0 = \frac{1}{1-\rho}$, $\mu_q^k(x) = (1-\rho)^kz_0^{|\supp x|}$ for $x\in\F_q^k$, and
$\rho\log_q\frac1{z_0} + \log_q\frac{1}{1-\rho} = h_q(\rho)$. Consequently, if $a_1,\dots,a_T\in\F_q^k$ are such that each $\ell_t(X):= \ip{a_t}{X}$ is $\rho$-sparse, then for any $z_1,\dots,z_T\in(0,1]$,
\begin{equation}\label{eq:chernoff}
H_q(X)\;\le\;k\,h_q(\rho)+\sum_{t=1}^T\rho\log_q\frac1{z_t}+\log_q\Eover{Y\sim\mu_q^k}{\prod_{t=1}^Tz_t^{\one\inbrak{\ell_t(Y)\neq0}}}.
\end{equation}
For $T = 0$, the bound in~\eqref{eq:chernoff} is $k\,h_q(\rho)$, and equality holds for $X\sim\mu_q^k$.
\end{lemma}
\begin{proof} We have $z_0<1$, since $\rho<1-\frac1q$, and each of the next three statements follows from $z_0(1-\rho) = \frac{\rho}{q-1}$ and the definition of $z_0$. We apply \eqref{eq:tilt} with $z_p = z_0$ for each of the $k$ coordinates and $z_p = z_t$ for each combination $\ell_t$. By the formula for $\mu_q^k(x)$,
\[
Z = \sum_{x\in\F_q^k}z_0^{|\supp x|}\prod_{t=1}^Tz_t^{\one\inbrak{\ell_t(x)\neq0}} = (1-\rho)^{-k}\,\Eover{\mu_q^k}{\prod_{t=1}^Tz_t^{\one\inbrak{\ell_t\neq0}}},
\]
and \eqref{eq:chernoff} follows from $\rho k\log_q\frac1{z_0}+\log_q(1-\rho)^{-k} = k\,h_q(\rho)$. For $T = 0$, the Gibbs measure is $z_0^{|\supp x|}/Z$, that is, $\mu_q^k$.
\end{proof}

\begin{lemma}\label{lem:activation} Let $\pi_w:= \frac{q-1}{q}\inparen{1-\kappa^w}$ for $w\ge0$, where $\kappa$ is as in Section~\ref{sec:sparsity}. Let $a\in\F_q^k$, let $w:=|\supp a|$, and let $Y\sim\mu_q^k$. Then the Fourier transform of $\mu_q^k$ satisfies $\widehat{\mu_q^k}(a) = q^{-k}\kappa^w$. Consequently
\begin{equation}\label{eq:activation-identity}
\PROver{\mu_q^k}{\ip{a}{Y} = v} = \begin{cases}
\frac1q\inparen{1+(q-1)\kappa^w}, & v = 0,\\[2pt]
\frac1q\inparen{1-\kappa^w}, & v\neq0,
\end{cases}
\qquad
\PROver{\mu_q^k}{\ip{a}{Y}\neq0} = \pi_w.
\end{equation}
Finally, $\pi_w$ is strictly increasing in $w\ge0$, $\pi_1 = \rho$ and $\pi_2 = \rho(1+\kappa)$.
In particular, $\gamma^* = \Dq{\rho}{\pi_2}$, and $\gamma^*>0$ since $\pi_2>\rho$.
\end{lemma}

For $v = 0$, \eqref{eq:activation-identity} is a special case of \cite[Lemma~5.6]{MRRSW} for vectors of length $w$, with $\beta = \rho$. Indeed, the $a_iY_i$ with $a_i\neq0$ are independent samples of $\mu_q$. For a single combination with support of size $2$, the minimum over $z_1$ of the right-hand side of \eqref{eq:chernoff} is $k\,h_q(\rho)-\gamma^*$, and it is attained at $z_1 = z^*$, for $z^*$ as in the following lemma.

\begin{lemma}\label{lem:tilt-opt} Let $z^*:= \frac{\rho(1-\pi_2)}{\pi_2(1-\rho)}$, and note that $z^*\in(0,1)$ since $\pi_2>\rho$. Let $a_1,\dots,a_T\in\F_q^k$ be such that each $\ell_t(X):= \ip{a_t}{X}$ is $\rho$-sparse, and let $M>0$ be such that
\[
\Eover{\mu_q^k}{\prod_{t=1}^T(z^*)^{\one\inbrak{\ell_t\neq0}}}\;\le\;M\inparen{1-\pi_2(1-z^*)}^T.
\]
Then $H_q(X)\le k\,h_q(\rho)-T\gamma^*+\log_q M$.
\end{lemma}

By Lemma~\ref{lem:activation}, $\Eover{\mu_q^k}{(z^*)^{\one\inbrak{\ell\neq0}}} = 1-\pi_w(1-z^*)$ for a $\rho$-sparse combination $\ell$ with support of size $w$, and $\pi_w\ge\pi_2$ for $w\ge2$.

\begin{corollary}\label{cor:single-form}
Let
$a\in\F_q^k$ with $|\supp a|\ge2$ be such that $\ip aX$ is also
$\rho$-sparse. Then $H_q(X)\le k\,h_q(\rho)-\gamma^*$.
\end{corollary}

The corollary follows from Lemma~\ref{lem:tilt-opt} with $T = 1$ and $M = 1$.

\subsection{General position}\label{sec:gp}
In this subsection we bound $H_q(X)$ using several combinations. We will see that the bound on $H_q(X)$ then tightens by $\gamma^*$ for each combination with support of size at least $2$, up to a term which is small when the combinations are in general position (see Lemma~\ref{lem:expansion}).
Throughout this subsection, let $a_1,\dots,a_T\in\F_q^k$ be nonzero vectors such that each $\ell_t(X):=\ip{a_t}{X}$ is $\rho$-sparse, and let $w_t:= |\supp a_t|$. For $c = (c_1,\dots,c_T)\in\F_q^T$ let $a_c:= \sum_{t=1}^T c_ta_t$ and $w(c):= |\supp a_c|$.
For $B\subseteq[T]$ we view
$c\in\F_q^B$ as an element of $\F_q^T$ by extending it by $0$.

\begin{definition}[General position]\label{def:gp}
Let $W\ge1$. We say that vectors $a_1,\dots,a_T\in\F_q^k$ are in \emph{$W$-general position} if $w(c)\ge W$ for all $c\in\F_q^T$ with $|\supp c|\ge2$, that is, if every $\F_q$-linear combination of $a_1,\dots,a_T$ with at least two nonzero coefficients has support of size at least $W$.
\end{definition}

\begin{lemma}\label{lem:expansion}
For any $z_1,\dots,z_T\in(0,1]$ let
$f_t(x):= z_t^{\one\inbrak{\ell_t(x)\neq0}}$. If $a_1,\dots,a_T$ are in $W$-general position, then
\begin{equation}\label{eq:expansion-gap}
\Eover{\mu_q^k}{\prod_t f_t}\;\le\;\inparen{1+q^T\kappa^W}\prod_{t=1}^T\Eover{\mu_q^k}{f_t}.
\end{equation}
\end{lemma}
\begin{proof}
Since $f_t = z_t + (1-z_t)\one\inbrak{\ell_t = 0}$ and $\one\inbrak{\ell_t(x) = 0} = \frac1q\sum_{c_t\in\F_q}\chi_{c_t}\inparen{\ell_t(x)}$ (see Section~\ref{sec:sparsity}), for any $x$ we have
\[
\prod_t f_t(x) = \sum_{B\subseteq[T]}\inparen{\prod_{t\notin B}z_t}
\inparen{\prod_{t\in B}\frac{1-z_t}{q}}
\sum_{c\in\F_q^B}\chi_{a_c}(x).
\]
By Lemma~\ref{lem:activation}, $\Eover{\mu_q^k}{\chi_{a_c}(Y)} = q^k\widehat{\mu_q^k}(-a_c) = \kappa^{w(c)}$. Hence
\[
\Eover{\mu_q^k}{\prod_{t=1}^Tf_t} =
\sum_{B\subseteq[T]}\inparen{\prod_{t\notin B}z_t}
\inparen{\prod_{t\in B}\frac{1-z_t}{q}}\sum_{c\in\F_q^B}\kappa^{w(c)}.
\]
In the inner sum, the terms with $|\supp c|\le1$ sum to $1+(q-1)\sum_{t\in B}\kappa^{w_t}\le\prod_{t\in B}\inparen{1+(q-1)\kappa^{w_t}}$. There are at most $q^T$ other terms, and each of them is at most $\kappa^W$, since $a_1,\dots,a_T$ are in $W$-general position. Hence the inner sum is at most $\inparen{1+q^T\kappa^W}\prod_{t\in B}\inparen{1+(q-1)\kappa^{w_t}}$, since this product is at least $1$. Finally, by \eqref{eq:activation-identity}, $\Eover{\mu_q^k}{f_t} = z_t+\frac{1-z_t}{q}\inparen{1+(q-1)\kappa^{w_t}}$, and writing the product of these expressions over $t\in[T]$ as a sum over $B\subseteq[T]$ gives \eqref{eq:expansion-gap}.
\end{proof}

We now choose the constants $T_0$ and $W_0$ so that for $T_0$ combinations in $W_0$-general position the bound on $H_q(X)$ tightens by more than $h_q(\rho)$, and $q^{T_0}\kappa^{W_0}/\ln q$ is at most $\gamma^*/2$. Let
\begin{equation}\label{eq:T0W0}
T_0:= \ceil{\frac{h_q(\rho)}{\gamma^*}}+1,\qquad
W_0:= \max\inset{2,\;\ceil{\log_{1/\kappa}\frac{2q^{T_0}}{\gamma^*\ln q}}}.
\end{equation}
Then $T_0\gamma^*\ge h_q(\rho)+\gamma^*$.

\begin{proposition}\label{prop:gp} Let $T_0$ and $W_0$ be as in~\eqref{eq:T0W0}, and suppose that $a_1,\dots,a_{T_0}\in\F_q^k$ are in $W_0$-general position, that each $a_t$ has support of size at least $2$, and that each $\ip{a_t}{X}$ is $\rho$-sparse. Then $H_q(X)<(k-1)\,h_q(\rho)$.
\end{proposition}
\begin{proof} Let $f_t:= (z^*)^{\one\inbrak{\ell_t\neq0}}$. By \eqref{eq:expansion-gap} and Lemma~\ref{lem:activation},
\[
\Eover{\mu_q^k}{\prod_tf_t}\;\le\;\inparen{1+q^{T_0}\kappa^{W_0}}\prod_t\inparen{1-\pi_{w_t}(1-z^*)},
\]
and $\pi_{w_t}\ge\pi_2$. By Lemma~\ref{lem:tilt-opt} with $M = 1+q^{T_0}\kappa^{W_0}$, and since $\log_q(1+u)\le u/\ln q$,
\[
H_q(X)\;\le\;k\,h - T_0\gamma^* +
\frac{q^{T_0}}{\ln q}\,\kappa^{W_0}.
\]
The last term is at most $\gamma^*/2$ by the choice of $W_0$, and $k\,h-T_0\gamma^*+\frac{\gamma^*}2<(k-1)\,h$ since $T_0\gamma^*\ge h+\gamma^*$.
\end{proof}

\subsection{Combinations with small support}\label{sec:support-reduction}
We can use Proposition~\ref{prop:gp} only when some $T_0$ of the combinations are in $W_0$-general position, and a configuration does not necessarily contain such combinations. In this subsection we show that otherwise, after a change of basis, many of the combinations have small support.

For $W\ge1$ let
$\LowW{W}:= \inset{a\in\F_q^k\suchthat|\supp a|<W}$.

\begin{lemma}\label{lem:dichotomy}
Let $T\ge2$, $W\ge1$, and let $\A\subseteq\F_q^k$ be nonempty. Then
$\A$ contains $T$ vectors in $W$-general position, or
$\A\subseteq V+\LowW{W}$ for some subspace $V\le\F_q^k$ with
$\dim V\le T-1$.
\end{lemma}
\begin{proof}
Let $a_1\in\A$ and $s:= 1$. As long as $s\le T-1$ and there is some $a\in\A$ such that $a_1,\dots,a_s,a$ are in $W$-general position, let $a_{s+1}$ be such an $a$, and increase $s$ by $1$. If $s$ becomes equal to $T$, then $a_1,\dots,a_T\in\A$ are $T$ vectors in $W$-general position, as desired. Otherwise let $V_s:=\spn(a_1,\dots,a_s)$, where $s\le T-1$. For any $a\in\A\setminus\inset{a_1,\dots,a_s}$ there exists a linear combination $c_0a+\sum_{i\le s}c_ia_i$ with $|\supp(c_0,\dots,c_s)|\ge2$ whose support has size less than $W$. In this combination $c_0\neq0$, because otherwise at least two of $c_1,\dots,c_s$ are nonzero, in contradiction to the fact that $a_1,\dots,a_s$ are in $W$-general position. Since $\LowW{W}$ is invariant under multiplication by nonzero field elements, it follows that $a\in c_0^{-1}\inparen{\LowW{W}-\sum_ic_ia_i}\subseteq V_s+\LowW{W}$. The $a_i$ themselves lie in $V_s$. Hence $\A\subseteq V_s+\LowW{W}$ with $\dim V_s\le T-1$.
\end{proof}

\begin{proposition}\label{prop:support-reduction}
Let $W\ge1$, let $X$ and $a_1,\dots,a_N$ be a configuration of dimension $k$ with $N\ge1$, and suppose
\[
\inset{a_j}_{j\in[N]}\subseteq V+\LowW{W}
\qquad\text{for a subspace } V\le\F_q^k.
\]
Then there is a change of basis of $\F_q^k$, under which at most one $e_i$ becomes a combination and at most one $a_j$ becomes a coordinate, such that the following holds. Let $\widetilde X$ be the vector of the new coordinates of $X$. Then $H_q(\widetilde X) = H_q(X)$, $\widetilde X$ and the new combinations are again a configuration, and there is a subset $J_0\subseteq[N]$ such that each $a_j$ with $j\in J_0$ is a combination of this configuration and
\[
|J_0|\;\ge\;\frac{N}{Wq^{\dim V}}-1 \qquad\text{and}\qquad
|\supp a_j|\;\le\;2W-1 \quad (j\in J_0),
\]
where the supports are taken in the new coordinates.
\end{proposition}
\begin{proof}
We write $a_j = u_j+v_j$ with $u_j\in V$ and $|\supp v_j|<W$. We first find a large set $F$ of indices $j$ with the same $u_j$. If the support of this vector has size at least $W$, we then change the basis so that the $a_j$ with $j$ in a large subset of $F$ have support of size at most $2W-1$ in the new coordinates.

Since $|V| = q^{\dim V}$, there are some $u^*\in V$ and a set $F$ of size at least $N/q^{\dim V}$ such that $u_j = u^*$ for all $j\in F$. If $|\supp u^*|<W$, then every $a_j$ with $j\in F$ has support of size at most $2W-2$, and we take $J_0:= F$ and the basis $\inset{e_i}_{i\in[k]}$.

Otherwise, let $A_1\subseteq\supp u^*$ be a set of size $W$. For $j\in F$ we have $|\supp v_j|<W$, and so $v_j$ is zero at some coordinate of $A_1$. Thus there are $i_0\in A_1$ and a set $F'\subseteq F$ of size at least $N/(Wq^{\dim V})$ such that $v_{j,i_0} = 0$ for all $j\in F'$.

Fix $j_2\in F'$. Then $a_{j_2,i_0} = u^*_{i_0}\neq0$, and so $\cB:= \inset{a_{j_2}}\cup\inset{e_i}_{i\neq i_0}$ is a basis for $\F_q^k$. Let $\widetilde X := (\ip bX)_{b\in\cB}$. Then $X\mapsto\widetilde X$ is an invertible linear map, and so $H_q(\widetilde X) = H_q(X)$. The vectors $a_{j_2}$ and $e_i$ with $i\neq i_0$ are the new coordinates, and the new combinations are given by $e_{i_0}$ and the $a_j$ with $j\neq j_2$. The coordinates of $\widetilde X$ are $\rho$-sparse, because each vector of $\cB$ is one of the $e_i$ and $a_j$. Two of the $e_i$ and $a_j$ are proportional in the new coordinates if and only if they are proportional in the coordinates of $X$. Thus $\widetilde X$ and the new combinations are again a configuration.

For $j\in F'$ let $d_j:= v_j-v_{j_2}$. Then $a_j = a_{j_2}+d_j$, $|\supp d_j|\le2W-2$ and $d_{j,i_0} = 0$. Hence $\ip{a_j}X = \ip{a_{j_2}}X+\sum_{i\neq i_0}d_{j,i}X_i$ is a linear combination of at most $2W-1$ coordinates of $\widetilde X$, and so $|\supp a_j|\le2W-1$ in the new coordinates. Finally, $J_0:= F'\setminus\inset{j_2}$ has size at least $N/(Wq^{\dim V})-1$.
\end{proof}

\subsection{Families with a common core}\label{sec:common-core}
In this subsection we prove two bounds on $H_q(X)$ for a family of combinations whose supports are the sets of a sunflower. We say that a set $U\subseteq[k]$ of coordinates is a \emph{core} of combinations $a_t = u_t+v_t$ if $\supp u_t\subseteq U$ for any $t$ and the \emph{petals} $P_t:= \supp v_t$ are pairwise disjoint and disjoint from $U$. Since the petals are disjoint, the combinations are independent under $\mu_q^k$ conditioned on the coordinates in the core. We prove the first bound for the combinations with $|P_t|\ge2$, and the second bound for the combinations with $|P_t| = 1$.

\begin{lemma}\label{lem:lambda}
Let $b\in\F_q^\times$, and, for $c\in\F_q$, let
\[
S(c):= \sum_{y\in\F_q}\inparen{\sqrt{z_0}}^{\one\inbrak{y\neq0}}
\inparen{\sqrt{z_0}}^{\one\inbrak{c+b y\neq0}}.
\]
Then $S(0) = 1+(q-1)z_0 = \frac1{1-\rho}$, $S(c) = 2\sqrt{z_0}+(q-2)z_0 = \frac{\lambda}{1-\rho}$ for every $c\neq0$, where $\lambda$ is as in Lemma~\ref{lem:delta}, and $\lambda<1$.
\end{lemma}
\begin{proof}
The first two statements follow by a straightforward computation of the terms, using $1+(q-1)z_0 = \frac1{1-\rho}$ (see Lemma~\ref{lem:calibration}). Finally, $\lambda<1$ is equivalent to $S(c)<S(0)$, that is, to $\inparen{1-\sqrt{z_0}}^2>0$.
\end{proof}

\begin{proposition}\label{prop:common-core}
Let $U\subseteq[k]$ be a core of the combinations
\[
\ell_t(X) = \ip{u_t+v_t}{X},\qquad t = 1,\dots,T,
\]
where $v_t\neq0$ and each $\ell_t(X)$ is
$\rho$-sparse. Let $T_{\ge2}$ be the number of $t$ with $|P_t|\ge2$.
\begin{enumerate}[label=\textup{(\roman*)}]
\item\label{item:core-large} The entropy of $X$ satisfies
\[
H_q(X)\;\le\;k\,h_q(\rho)-T_{\ge2}\,\gamma^*.
\]
\item\label{item:core-single} If $|P_t| = 1$ for every $t$, and some $i^*\in U$ has $(u_t)_{i^*}\neq0$ for every $t$, then
\[
H_q(X)\;\le\;(k-1)\,h_q(\rho)\;+\;\log_q\inparen{1+(q-1)\lambda^{\,T}}.
\]
\end{enumerate}
\end{proposition}

\begin{proof} For part~\ref{item:core-large}, we apply Lemma~\ref{lem:tilt-opt} to the combinations with $|P_t|\ge2$, and we ignore the combinations with $|P_t| = 1$. Let $Y$ be a random vector sampled according to $\mu_q^k$. Note that the values $c_t:= \ip{u_t}{Y}$ depend only on the coordinates of $Y$ that are not in $\bigsqcup_tP_t$, because $\supp u_t\subseteq U$ and $U$ is disjoint from every $P_t$. Since the $P_t$ are disjoint, the $\ell_t(Y)$, conditioned on the values of these coordinates, are independent, and by \eqref{eq:activation-identity} the expectation corresponding to $t$ is
\[
z^*+(1-z^*)\,\PROver{\mu_q^k}{\ip{v_t}{Y_{P_t}} = -c_t}\;\le\;z^*+(1-z^*)\,\PROver{\mu_q^k}{\ip{v_t}{Y_{P_t}} = 0}\;=\;1-\pi_{|P_t|}(1-z^*).
\]
Part~\ref{item:core-large} follows from $\pi_{|P_t|}\ge\pi_2$ for $|P_t|\ge2$ and Lemma~\ref{lem:tilt-opt} with $M = 1$.

For part~\ref{item:core-single}, we use $z_p = z_0$ for every coordinate of $[k]\setminus\inset{i^*}$ that is not in $\bigsqcup_tP_t$, and $z_p = \sqrt{z_0}$ for each of the $T$ coordinates in $\bigsqcup_tP_t$ and each combination. All these $z_p$ are in $(0,1]$, and all these constraints are $\rho$-sparse, so we may apply Fact~\ref{fact:gibbs}. The first term of \eqref{eq:tilt} is $\rho\log_q\frac1{z_0}\,(k-1-T)+2T\rho\log_q\frac1{\sqrt{z_0}} = \rho\,(k-1)\log_q\frac1{z_0}$.

For $Z$, fix the coordinates of $x$ that are not in $\bigsqcup_tP_t\cup\inset{i^*}$, and let $s:= x_{i^*}$. The sum over the coordinates in $\bigsqcup_tP_t$ is the product over $t$ of the sums over the coordinate in $P_t$, because the $P_t$ are disjoint. By Lemma~\ref{lem:lambda}, the sum over the coordinate in $P_t$ is $S(\ip{u_t}{x})$. Since $(u_t)_{i^*}\neq0$, we have $\ip{u_t}{x} = 0$ for exactly one value $s = s_t$ of $x_{i^*}$. Let $\cT_s:= \inset{t\suchthat s_t = s}$. Then the sum over $s$ and over the coordinates in $\bigsqcup_tP_t$ is
\[
(1-\rho)^{-T}\sum_{s\in\F_q}\lambda^{T-|\cT_s|}\;\le\;(1-\rho)^{-T}\inparen{1+(q-1)\lambda^T},
\]
where the inequality holds because $y\mapsto\lambda^{T-y}$ is convex and $\sum_s|\cT_s| = T$. By Lemma~\ref{lem:calibration}, the sum over each of the other $k-1-T$ coordinates is $\frac1{1-\rho}$. Thus
$Z\le(1-\rho)^{-(k-1)}\inparen{1+(q-1)\lambda^T}$, and part~\ref{item:core-single} follows by Lemma~\ref{lem:calibration}.
\end{proof}

\subsection{Sunflowers}\label{sec:endgame}
By Proposition~\ref{prop:support-reduction} with $W = W_0$, after a change of basis there is a large subset of the combinations whose supports have size at most $\bar W:= 2W_0-1$. In this subsection we find a sunflower of these supports and apply Proposition~\ref{prop:common-core} to the corresponding combinations.

\begin{fact}[Sunflower lemma, \cite{ER60}]\label{fact:sunflower} Let $T\ge2$, $b\ge0$, and let $\cF$ be a family of distinct sets, each of size at most $b$. If $|\cF|>b!\,(T-1)^b$ then $\cF$ contains a \emph{$T$-sunflower}, that is, sets $S_1,\dots,S_T\in\cF$ with $S_i\cap S_j = U$ for all $i\neq j$, for some set $U$ called the \emph{core}.
\end{fact}

This statement can be proved by the standard inductive argument (see e.g.\ \cite[Chapter~6]{Jukna11}). 

In Proposition~\ref{prop:sunflower-bound} we use a $T(\eta)$-sunflower, where for $\eta>0$
\begin{equation}\label{eq:target-count}
T(\eta):= 1+2\max\inset{\ceil{\frac{h}{\gamma^*}},\;\ceil{g_\eta}},
\qquad g_\eta:= \log_{1/\lambda}\frac{q-1}{\eta\ln q}.
\end{equation}
Note that $T(\eta) = O_{q,\rho}(1+\log(1/\eta))$. We choose $T(\eta)$ so that if at least half of the petals have size at least $2$, then by part~(i) of Proposition~\ref{prop:common-core} the bound on $H_q(X)$ tightens by at least $h_q(\rho)$, and if at least half of them have size $1$, then by part~(ii), $H_q(X)\le(k-1)\,h_q(\rho)+\eta$.

\begin{proposition}\label{prop:sunflower-bound} Let $X$ and $a_1,\dots,a_N$ be a configuration of dimension $k$, let $\eta>0$, let $b\ge0$ be an integer, and suppose that $J_0\subseteq[N]$ satisfies
\[
|\supp a_j|\;\le\;b\quad(j\in J_0)
\qquad\text{and}\qquad
|J_0|\;>\;(q-1)^b\,b!\,\bigl(T(\eta)-1\bigr)^b.
\]
Then $H_q(X)\;\le\;(k-1)\,h_q(\rho) + \eta$.
\end{proposition}
\begin{proof}
A combination with support $S$ has a nonzero coefficient at each coordinate of $S$, so at most $(q-1)^{|S|}\le(q-1)^b$ distinct combinations have support $S$, and the number of distinct supports of the combinations in $J_0$ is larger than $b!\,(T(\eta)-1)^b$. By Fact~\ref{fact:sunflower}, there is a $T(\eta)$-sunflower of supports $S_1,\dots,S_{T(\eta)}$ with core $U$. Let $P_t:= S_t\setminus U$. Since the $S_t$ are distinct, $P_t = \emptyset$ for at most one $t$, and so there are at least $T(\eta)-1$ sets $P_t$ with $P_t\neq\emptyset$. We choose one combination for each $t$ with $P_t\neq\emptyset$, and write it as $u_t+v_t$, where $\supp u_t\subseteq U$ and $\supp v_t = P_t\neq\emptyset$.

Since $S_i\cap S_j = U$ for all $i\neq j$, we have $U\subseteq S_t$ for any $t$, and each $u_t$ has support exactly $S_t\cap U = U$. The sets $P_t$ are disjoint from each other and from $U$, so we may apply Proposition~\ref{prop:common-core} with core $U$.

If at least $\frac{T(\eta)-1}2\ge\ceil{\frac{h}{\gamma^*}}$ of the $P_t$ with $P_t\neq\emptyset$ have size at least $2$, then by Proposition~\ref{prop:common-core}\ref{item:core-large}, $H_q(X)\le k\,h-\ceil{\frac{h}{\gamma^*}}\gamma^*\le(k-1)h$.

Otherwise, at least $\frac{T(\eta)-1}2$ of the $P_t$ with $P_t\neq\emptyset$ have size $1$. Let $T_1$ denote the number of these $P_t$. For such a $t$, we have $u_t\neq0$. Indeed, if $u_t = 0$ then the combination $u_t+v_t = v_t$ is a nonzero multiple of some $e_i$, in contradiction to the fact that $e_1,\dots,e_k,a_1,\dots,a_N$ are pairwise non-proportional. In particular, $U\neq\emptyset$, and any $i^*\in U$ has $(u_t)_{i^*}\neq0$ for every such $t$. We apply Proposition~\ref{prop:common-core}\ref{item:core-single} to the $T_1$ combinations whose petal has size $1$. By \eqref{eq:target-count}, we have $T_1\ge\frac{T(\eta)-1}2\ge g_\eta$, and since $\log_q(1+u)\le u/\ln q$, we conclude that
\[
H_q(X)\;\le\;(k-1)h+\frac{q-1}{\ln q}\,\lambda^{g_\eta}
\;=\;(k-1)h+\eta.
\]
\end{proof}

\subsection{Proof of \texorpdfstring{Theorem~\ref{thm:tool}}{Theorem 1.6}}\label{sec:low}

\begin{proof}[Proof of Theorem~\ref{thm:tool}]
Part~(1) follows from $|\supp a_1|\ge2$ and Corollary~\ref{cor:single-form}.

For part~(2), let $N_0 = N_0(q,\rho,\eta)$ be the minimum integer such that
\begin{equation}\label{eq:choice-N0}
\frac{N_0}{W_0\,q^{T_0-1}}-1\;>\;
(q-1)^{\bar W}\,\bar W!\,\bigl(T(\eta)-1\bigr)^{\bar W}.
\end{equation}
Since $T(\eta) = O_{q,\rho}(1+\log(1/\eta))$ by \eqref{eq:target-count}, $N_0$ is polynomial in $\log(1/\eta)$ of degree $\bar W$. Suppose that $N\ge N_0$. We apply Lemma~\ref{lem:dichotomy} to $\A:= \inset{a_j}_{j\in[N]}$ with $T:= T_0$ and $W:= W_0$.

Consider first the case where some $T_0$ of the $a_j$ are in $W_0$-general position. Then by Proposition~\ref{prop:gp},
$H_q(X)\le(k-1)\,h\le(k-1)\,h+\eta$.

Otherwise, there is a subspace $V\le\F_q^k$ with $\dim V\le T_0-1$ such that $\inset{a_j}_{j\in[N]}\subseteq V+\LowW{W_0}$.
By Proposition~\ref{prop:support-reduction},
after a change of basis there is a subset $J_0$ with
\[
|J_0|\;\ge\;\frac{N}{W_0\,q^{\dim V}}-1
\;\ge\;\frac{N_0}{W_0\,q^{T_0-1}}-1
\qquad\text{and}\qquad
|\supp a_j|\;\le\;2W_0-1 = \bar W\quad(j\in J_0),
\]
where the supports are taken in the new coordinates. By
\eqref{eq:choice-N0}, we may apply Proposition~\ref{prop:sunflower-bound} to $\widetilde X$ and the new combinations, with $J_0$ and $b = \bar W$. Since $H_q(\widetilde X) = H_q(X)$, we obtain $H_q(X)\le(k-1)\,h+\eta$.
\end{proof}

\section{Proof of \texorpdfstring{Proposition~\ref{prop:sigma-min}}{Proposition 3.3}}\label{sec:lower}
In this section we prove Proposition~\ref{prop:sigma-min}. \cite[Lemma~4.4]{GLMRSW} estimate $\Rs{\sigma_m}$, and here we calculate it exactly using symmetry and submodularity.

Let $Y = u+\alpha\onesvector$ be as in Definition~\ref{def:sigma}, and for a subspace $V\le\F_q^m$ let $G(V):= H_q(Y\bmod V)$. The function $G$ is submodular, that is,
\[
G(V_1\cap V_2)+G(V_1+V_2)\;\le\;G(V_1)+G(V_2).
\]
Indeed, $Y\bmod(V_1\cap V_2)$ and $(Y\bmod V_1,Y\bmod V_2)$ are in one-to-one correspondence, and $Y\bmod(V_1+V_2)$ is a function of $Y\bmod V_1$ and also of $Y\bmod V_2$. Hence, $G(V_1)+G(V_2)-G(V_1\cap V_2) = I(Y\bmod V_1;Y\bmod V_2)\ge G(V_1+V_2)$. Note that $G$ is invariant under the group $S_m$ of permutations of the coordinates, since $\sigma_m$ is invariant under these permutations.

Let $S_0:= \inset{a\in\F_q^m\suchthat\sum_ja_j = 0}$ and $Q_m:= \frac{m\,h_q(\rho)-\Delta_m}{m-1}$.
Since $\alpha\onesvector\in\F_q\onesvector$, we have $G(\F_q\onesvector) = H_q(u\bmod\F_q\onesvector) = (m-1)\,Q_m$ by the definition of $\Delta_m$ in Lemma~\ref{lem:delta}. Note that $Q_m\le1$, because $u\bmod\F_q\onesvector$ takes at most $q^{m-1}$ values.

\begin{proof}[Proof of Proposition~\ref{prop:sigma-min}] The span of the support of $\sigma_m$ is $\F_q^m$. Thus $1-\Rs{\sigma_m}$ is the minimum of $G(V)/(m-\dim V)$ over $V\lneq\F_q^m$. Let $c$ be this minimum, and let $V^*$ be a subspace of minimum dimension such that the minimum $c$ is attained at $V^*$. We first show that $\pi V^* = V^*$ for every $\pi\in S_m$. Let $W:= V^*\cap\pi V^*$. We have $G(V^*+\pi V^*)\ge c\,(m-\dim(V^*+\pi V^*))$, and this also holds when $V^*+\pi V^* = \F_q^m$, since in this case the left-hand side and the right-hand side are $0$. By submodularity and $G(\pi V^*) = G(V^*)$,
\[
G(W)\;\le\;2\,G(V^*)-G(V^*+\pi V^*)\;\le\;c\inparen{2(m-\dim V^*)-(m-\dim(V^*+\pi V^*))}\;=\;c\,(m-\dim W).
\]
Thus the minimum is also attained at $W\le V^*$, and by the definition of $V^*$ we have $W = V^*$, that is, $\pi V^* = V^*$.

Next, if $V^*$ contains some $a$ that does not belong to $\F_q\onesvector$, then $a_i\neq a_j$ for some $i\neq j$, and for the transposition $(ij)$ we have $a-(ij)a = (a_i-a_j)(e_i-e_j)\in V^*$. Thus $V^*$ contains every $e_{i'}-e_{j'}$, and so $S_0\le V^*$. Since $S_0$ has dimension $m-1$, we then have $V^* = S_0$. If $V^*$ does not contain any such $a$, then $V^*\le\F_q\onesvector$. Hence $V^*$ is one of $0$, $\F_q\onesvector$ and $S_0$, and it suffices to show that $G(0)\ge m\,Q_m$ and $G(S_0)\ge Q_m$.

For $V = 0$, note that $Y\bmod\F_q\onesvector = u\bmod\F_q\onesvector$ is a function of $u$. Thus
\[
H_q(Y) = H_q(Y\bmod\F_q\onesvector)+H_q(Y\mid Y\bmod\F_q\onesvector) \ge (m-1)\,Q_m+H_q(Y\mid u) = (m-1)\,Q_m+1 \ge m\,Q_m,
\]
since $H_q(Y\mid u) = H_q(\alpha) = 1$ and $Q_m\le1$. For $V = S_0$, $Y\bmod S_0$ is in one-to-one correspondence with $\sum_jY_j = \sum_ju_j+m\alpha$. Hence
\[
G(S_0)\;\ge\;H_q\Bigl(\sum_jY_j\Bigm|\alpha\Bigr)\;=\;H_q\Bigl(\sum_ju_j\Bigr)\;\ge\;H_q\Bigl(\sum_ju_j\Bigm|u_3,\dots,u_m\Bigr)\;=\;H_q(u_1+u_2).
\]
On the other hand, $u\bmod\F_q\onesvector$ is in one-to-one correspondence with $(u_j-u_m)_{j<m}$, and each $u_j-u_m$ is distributed identically to $u_1+u_2$, because $\mu_q$ is invariant under multiplication by $-1$. Thus $(m-1)\,Q_m\le(m-1)\,H_q(u_1+u_2)$, and so $G(S_0)\ge Q_m$. Therefore $c = Q_m$, and the minimum is attained at $V = \F_q\onesvector$.
\end{proof}

\section{Random LDPC codes}\label{sec:ldpc}

In this section we prove Theorem~\ref{thm:ldpc-same} and Corollary~\ref{cor:ldpc-list}.

\begin{theorem}\label{thm:ldpc-core}
Let $\tau$ be a scale-invariant distribution on $\F_q^m$ with full support, and let $\gamma>0$ and $R_0<1$. Let $\tau_n$ be a scale-invariant rounding of $\tau$ to denominator $n$ such that, when $q = 2$, the matrices of type $\tau_n$ have columns of even weight. There is $s_0 = s_0(q,\tau,\gamma,R_0)$ such that for any odd $s\ge s_0$ and any rate $R$ with $\Rs\tau+\gamma\le R\le R_0$, a random $s$-LDPC code of rate $R$ contains a matrix of type $\tau_n$ with probability $1-o(1)$.
\end{theorem}

In Section~\ref{sec:ldpc-single} we derive Theorem~\ref{thm:ldpc-same} and Corollary~\ref{cor:ldpc-list} from Theorem~\ref{thm:ldpc-core}. We then prove Theorem~\ref{thm:ldpc-core} by the second moment method, following the outline in Section~\ref{sec:overview-ldpc}. In Section~\ref{sec:ldpc-blocks} we compute the probability that the code contains a given matrix, and we restrict our attention to codes of good distance. In Section~\ref{sec:ldpc-cycles} we study the short cycles of the Tanner graph, and in Section~\ref{sec:ldpc-assembly} we complete the proof. We prove Lemmas~\ref{lem:ldpc-layer}, \ref{lem:ldpc-smooth}, \ref{lem:ldpc-planted} and~\ref{lem:ldpc-cycles} in Appendix~\ref{app:ldpc}.
Throughout this section and Appendix~\ref{app:ldpc}, $H$ denotes entropy with the natural logarithm,
$H_q = H/\ln q$, and $D$ denotes the Kullback--Leibler divergence with the natural logarithm.

\subsection{Proof of \texorpdfstring{Theorem~\ref{thm:ldpc-same} and Corollary~\ref{cor:ldpc-list}}{Theorem 1.4 and Corollary 1.5}}\label{sec:ldpc-single}

\begin{proof}[Proof of Theorem~\ref{thm:ldpc-same}]
The first statement follows from \cite[Theorem~1.9]{MRRSW}. For the second statement, we find a scale-invariant type in $\cT$ whose threshold rate is close to $\TRLC$. Choose $\tau\in\cT$ with $\Rs\tau\le\TRLC+\frac\gamma3$, let $C$ be the span of its support, and let $r:=\dim C$. By condition~(2), some $\tau'\in\cT$ has support $C$. For $\theta := \frac\gamma3$, let $\nu$ be the distribution of $\xi x$, where $x\sim(1-\theta)\tau+\theta\tau'$ and $\xi$ is uniform on $\F_q^\times$ and independent of $x$. By condition~(1), $\nu\in\cT$, and $\nu$ is scale invariant and has support $C$. For each $V\le C$, the implied type $\nu\bmod V$ is a mixture of the distributions of $cx+V$, where $c\in\F_q^\times$ and $x\sim\tau$ or $x\sim\tau'$. For $x\sim\tau$, the distribution of $cx+V$ has the same entropy as $\tau\bmod V$, because $cV = V$. By concavity of entropy, $H_q(\nu\bmod V)\ge(1-\theta)\,H_q(\tau\bmod V)$. Hence $\Rs\nu\le\Rs\tau+\theta$, since $H_q(\tau\bmod V)\le r-\dim V$.

We now project $\nu$ onto a set of coordinates, since Theorem~\ref{thm:ldpc-core} is stated for distributions of full support. Choose a set $I$ of $r$ coordinates on which the projection $C\to\F_q^I$ is
injective. The projection $\nu_I$ of $\nu$ is a scale-invariant distribution of full support on $\F_q^I\cong\F_q^r$, and $\Rs{\nu_I} = \Rs\nu$. By condition~(3), $\nu$ has scale-invariant
roundings $\nu_n$ that are bad types with support in $C$. Their projections are scale-invariant roundings of $\nu_I$ and, when $q = 2$, their matrices have columns of even weight. Note that $R\ge\TRLC+\gamma\ge\Rs{\nu_I}+\frac\gamma3$. Theorem~\ref{thm:ldpc-core} with $\frac\gamma3$ in place of $\gamma$ and $R_0 = 1-\gamma$ implies that with probability $1-o(1)$ the code contains a matrix $B$ whose type is the projection of $\nu_n$. Each row of $B$ is the projection of a unique vector of $C$, and the columns of the matrix $B'$ whose rows are these vectors are linear combinations of the columns of $B$. By linearity, the code also contains the matrix $B'$, and its type is $\nu_n$, which is bad.
\end{proof}

We now check conditions~(1)--(3) for $(\rho,m-1)$-list-decodability. We take $\cT$ to be the $\rho$-clustered types of length $m$, and by Lemma~\ref{fact:threshold}, $\TRLC = \TRLC(\rho,m-1)$. Condition~(1) holds. Indeed, $\cT$ is closed under mixtures of two types, and under multiplication of the rows of a type and of its center by one nonzero element of $\F_q$. For condition~(2), let $\tau\in\cT$
have center $\alpha$ and span $C$. Let $\tau'$ be the distribution which is $\tau$ with probability $1-\theta_1-\theta_2$, the row $0$ with center $0$ with probability $\theta_1$, and a uniformly random row in $C$ with center $0$ with probability $\theta_2$, where $\theta_2\le \theta_1\rho/(1-\rho)$. Then each coordinate is different from its center with probability at most $(1-\theta_1-\theta_2)\rho+\theta_2\le\rho$, any two coordinates are still different with positive probability, and the support is all of $C$.

For condition~(3), let $\nu\in\cT$ be scale invariant with support $C$. We may assume that the distribution of the pair $(v,c)$ of a row $v$ and its center $c$ is scale invariant. Indeed, otherwise we take the distribution of $(\xi v,\xi c)$ in place of it, where $\xi$ is uniform on $\F_q^\times$. Let $\nu_n$ be the distribution of $v'$, where $(v',c')$ is a random pair such that, for each $(x,y)\neq(0,0)$, the probability that $(v',c') = (x,y)$ is the probability that $(v,c) = (x,y)$, rounded down to a multiple of $\frac2n$, and the probability that $(v',c') = (0,0)$ is such that the probabilities sum to $1$. As in the proof of Lemma~\ref{fact:threshold}, for $n$ large enough $\nu_n$ is a bad type with support in $C$. Note that $\nu_n$ is scale invariant, since equal probabilities are still equal after rounding down, and $\xi(0,0) = (0,0)$ for every $\xi\in\F_q^\times$. Every nonzero row appears an even number of times in a matrix of this type. In particular, for $q = 2$ every column has even weight.

Similarly, distance, average-radius list decoding and list recovery from errors satisfy conditions~(1)--(3). For distance $\delta$ we take $m = 1$ and let $\cT$ be the set of distributions $\tau$ on $\F_q$ with $\tau(0)>1-\delta$ and $\tau\ne\mathbf 1_0$, and for condition~(2) we replace $\tau$ by the distribution which is $\tau$ with probability $1-\theta$ and uniform on $\F_q$ with probability $\theta$, for a small $\theta>0$. For average-radius list decoding we require only that the average over the $m$ columns of the probability that a coordinate is different from its center is at most $\rho$. The distributions that we define in the proofs of conditions~(2) and~(3) also satisfy this condition. For list recovery from errors each row has a list of $\ell$ symbols in place of a center, and each column is not in its list with probability at most $\rho$. The rows that we add in the proofs of conditions~(2) and~(3) have a list containing $0$. In the proofs of condition~(2) for list decoding, average-radius list decoding and list recovery we need $\rho>0$, because for $\rho = 0$ we must take $\theta_2 = 0$.

\begin{proof}[Proof of Corollary~\ref{cor:ldpc-list}]
The first statement follows from Theorem~\ref{thm:ldpc-same} applied to list decoding. For a fixed $\eps$, the proofs of Theorem~\ref{thm:main-twopoint} and Corollary~\ref{cor:all-eps} in Section~\ref{sec:derive} show that $R<\TRLC(\rho,L)$ for one value of $L$, and that $R>\Rs{\sigma_m}\ge\TRLC(\rho,m-1)$ for at most two values of $m$. Hence the second statement follows from the first, with $\gamma$ smaller than the gaps in these inequalities.
\end{proof}

\subsection{Lemmas for the second moment}\label{sec:ldpc-blocks}

In what follows, let $\tau$, $\gamma$, $R$ and $R_0$ be as in Theorem~\ref{thm:ldpc-core}, let $\cC$ be a random $s$-LDPC code of rate $R$, and let $X$ be the number of matrices of type $\tau_n$ in $\cC$. To apply the second moment method to $X$, we first compute the probability that $\cC$ contains a given matrix.

For random linear codes, the probability that a matrix is contained in the code depends only on its rank, but for random LDPC codes it depends on the type of the matrix. The distribution of a random $s$-LDPC code is invariant under permutations of the coordinates. Hence the probability that a matrix $B\in\F_q^{n\times d}$ satisfies the parity checks of one layer depends only on the type $\omega$ of $B$, and we denote it by $p(\omega)$. Since the $t$ layers are independent, $B\subseteq\cC$ with probability $p(\omega)^t$.

A matrix $B$ satisfies the parity checks of one layer if and only if, for every block, the rows of $B$ in that block satisfy the parity check of that block. Let $\mu$ be a distribution on $\F_q^d$. Let $w = (w_1,\dots,w_s)$ be a block whose rows are independent samples from $\mu$, and let $\xi_1,\dots,\xi_s$ be independent and uniformly random in $\F_q^\times$. Then $\sum_i\xi_iw_i = 0$ with probability
\[
P_d(\mu) := \sum_{w\in(\F_q^d)^s}\PR{\textstyle\sum_i\xi_iw_i=0}\prod_{i=1}^s\mu_{w_i}.
\]
Let $\omega_\mu$ be the distribution of a row $w_1$ of the block, conditioned on $\sum_i\xi_iw_i = 0$, and let $\pi_\mu$ be the joint distribution of two rows $(w_1,w_2)$, conditioned on the same event.

As in the proof of \cite[Lemma~2.13]{MRRSW}, we compute $p(\omega)$ using a random matrix $\Lambda\in\F_q^{n\times d}$ whose rows are independent samples of $\mu$. Let $S_\mu$ be $n$ times the type of $\Lambda$, conditioned on $\Lambda$ satisfying the parity checks of one layer. Conditioned on this event, the blocks of $\Lambda$ are independent. Thus $S_\mu$ is a sum of $n/s$ independent random vectors.

\begin{lemma}\label{lem:ldpc-layer} Let $\omega$ be the type of a matrix in $\F_q^{n\times d}$, and let $\mu$ be a distribution on $\F_q^d$ with $\mu_a>0$ whenever $\omega_a>0$. Then $p(\omega)$ is the probability that $\Lambda$ satisfies the parity checks of one layer, conditioned on $\Lambda$ having type $\omega$. In particular,
\begin{equation}\label{eq:ldpc-bayes}
p(\omega)\;=\;\frac{P_d(\mu)^{n/s}\,\PR{S_\mu = n\omega}}{\binom n{n\omega}\prod_a\mu_a^{n\omega_a}}.
\end{equation}
\end{lemma}

As in \cite[Section~4]{MRRSW}, we compute $P_d(\mu)$ using the Fourier transform. For a distribution $\mu$ on $\F_q^d$ and $z\in\F_q^d$ let
\[
(K\mu)(z) := \Eover{v\sim\mu,\,\xi}{\chi_z(\xi v)} = 1-\frac q{q-1}\PR{\ip zv\ne0},
\]
where $\xi$ is uniform on $\F_q^\times$, and the second equality follows as in the proof of \cite[Claim~4.4]{MRRSW}. We extend the definition of $K$ to all functions on $\F_q^d$ by linearity. In the notation of \cite[Section~4]{MRRSW}, $K\mu = q^d\hat{\mu'}$, where $\mu'$ is the distribution of $\xi v$ for $v\sim\mu$. In terms of the Fourier transform, we have $P_d(\mu) = q^{-d}\sum_z((K\mu)(z))^s$ (see Section~\ref{sec:sparsity}).

As in \cite{MRRSW}, we say that $\mu$ is \emph{$\delta$-smooth} if $\PROver{v\sim\mu}{\ip zv\ne0}\ge\delta$ for all $z\ne0$. When $q = 2$ we also require that $\PR{\ip zv=0}\ge\delta$. Equivalently, the type of a matrix $B$ is $\delta$-smooth if and only if $B$ has full rank and every nonzero codeword of the code spanned by its columns has relative weight at least $\delta$, and, when $q = 2$, at most $1-\delta$ (see \cite{MRRSW}). The probability that $B$ satisfies a parity check is close to $q^{-d}$ only if the type of $B$ is smooth. If $\mu$ is $\delta$-smooth, then by the formula for $K\mu$ above, $|(K\mu)(z)|\le\kappa_\delta<1$
for $z\ne0$, where $\kappa_\delta := \max\inset{1-\frac{q\delta}{q-1},\frac1{q-1}}$ for $q\ge3$ and
$\kappa_\delta := 1-2\delta$ for $q = 2$. Hence the terms of $P_d(\mu)$ with $z\neq0$ are exponentially small in $s$.

If $\mu$ is smooth, then conditioned on the parity check of a block, each row of the block has distribution close to $\mu$, and two rows of the block are close to independent (Lemma~\ref{lem:ldpc-smooth}(1)). For every smooth scale-invariant $\omega$ there is also a distribution $\mu^*$ close to $\omega$ with $\omega_{\mu^*} = \omega$ (Lemma~\ref{lem:ldpc-smooth}(3)). For $\delta>0$ let $\eta_\delta := 8q^{2m}\kappa_{\delta/2}^{\,s-2}$. Let $\Sigma(\omega) := \diag(\omega)-\omega\omega^{\top}$, and let $A_\mu := (s-1)\,\Sigma(\omega_\mu)^{-1}\inparen{\pi_\mu-\omega_\mu\omega_\mu^{\top}}$ on the space of functions on $\F_q^d$ modulo constants.

\begin{lemma}\label{lem:ldpc-smooth} Let $d\le2m$, and let $\delta>0$ be such that $s\eta_\delta\le\frac18$.
\begin{enumerate}[label=\textup{(\arabic*)}]
\item If $\mu$ is $\frac\delta2$-smooth, then $\big|\frac{\omega_\mu(a)}{\mu_a}-1\big|\le\frac{\eta_\delta}2$ and $\inabs{\pi_\mu(a,b)-\omega_\mu(a)\omega_\mu(b)}\le\eta_\delta\,\mu_a\mu_b$ for all $a,b$.
\item If $\mu$ is $\frac\delta2$-smooth, then $s$ times the type of $w$, conditioned on $\sum_i\xi_iw_i = 0$, has covariance $s\,\Sigma(\omega_\mu)(I+A_\mu)$, and $(1-s\eta_\delta)\,\Sigma(\omega_\mu)\preceq\Sigma(\omega_\mu)(I+A_\mu)\preceq(1+s\eta_\delta)\,\Sigma(\omega_\mu)$.
\item If $\omega$ is $\delta$-smooth and scale invariant, then there is a scale-invariant $\frac\delta2$-smooth distribution $\mu^*$ with $\omega_{\mu^*} = \omega$ and $|\mu^*_a/\omega_a-1|\le\eta_\delta$ for all $a$.
\end{enumerate}
\end{lemma}

Let $\delta_\tau>0$ be the largest $\delta$ such that $\tau$ is $\delta$-smooth, let $\tau_{\min} := \min_v\tau(v)$, and let $\bar\kappa := \kappa_{\delta_\tau/4}$, so that $\eta_{\delta_\tau/2} = 8q^{2m}\bar\kappa^{\,s-2}$. By Lemma~\ref{lem:ldpc-smooth}(3) with $\delta = \delta_\tau$, there is a scale-invariant distribution $\nu$ with $\omega_\nu = \tau$ and $|\nu_a/\tau_a-1|\le\eta_{\delta_\tau}$ for all $a$. In other words, if the rows of a block are independent samples of $\nu$, then each row, conditioned on the block satisfying its parity check, has distribution $\tau$.

\begin{lemma}\label{lem:ldpc-one}
A matrix of type $\tau_n$ is contained in $\cC$ with probability
\[
p(\tau_n)^t\;=\;q^{-(1-R)mn+O(n(1-R)q^m\bar\kappa^{\,s})+O(\log n)}.
\]
In particular, $\E X$ is the expected number of matrices of type $\tau_n$ in a random linear code of the same rate, up to a $q^{O(n(1-R)q^m\bar\kappa^s)+O(\log n)}$ factor.
\end{lemma}

\begin{proof}
By~\eqref{eq:ldpc-bayes} with $\mu = \nu$ and by Stirling's formula, $p(\tau_n) = e^{nG(\tau)+O(\log n)}\PR{S_\nu = n\tau_n}$, where $G(\tau) := \frac1s\ln P_m(\nu)+D(\tau\,\|\,\nu)$. Since $\omega_\nu = \tau$, the expectation of $S_\nu$ is $n\tau$, and so $\PR{S_\nu = n\tau_n} = n^{-O(1)}$ by the local limit theorem (Fact~\ref{fact:llt}). By the formula for $P_m(\nu)$ in terms of the Fourier transform, and since $|\nu/\tau-1|\le\eta_{\delta_\tau}$, we have $tG(\tau) = -(1-R)m\ln q+O\inparen{(1-R)q^m\bar\kappa^{\,s}}$. The last statement follows since a random linear code of rate $R$ contains a fixed one of them with probability $q^{-n(1-R)m}$.
\end{proof}

As in Section~\ref{sec:overview-ldpc}, we restrict our attention to codes of good distance. Let $D_\delta$ be the event
that every nonzero codeword of $\cC$ has weight at least $\delta n$. For a random variable $F$, let $\Epl[F] := \E{XF}/\E X$. Note that $\Epl[F]$ is the average over the matrices $M$ of type $\tau_n$ of the expectation of $F$ conditioned on $M\subseteq\cC$, since each of these matrices is contained in $\cC$ with probability $p(\tau_n)^t$.

\begin{lemma}\label{lem:ldpc-planted} Let $\delta>0$ be small enough in terms of $q$, $\delta_\tau$ and $R_0$. Then for $s$ large enough, a random $s$-LDPC code that contains a given matrix of type $\tau_n$ has no nonzero codeword of weight less than $\delta n$ with probability $1-o(1)$, that is, $\Epl[\one_{D_\delta^c}] = o(1)$.
\end{lemma}

Fix such a $\delta$, and let $Y := X\one_{D_\delta}$. By Lemma~\ref{lem:ldpc-planted}, $\E Y = (1-o(1))\E X$.

\subsection{Short cycles}\label{sec:ldpc-cycles}

As in Section~\ref{sec:overview-ldpc}, $Y$ depends on the numbers of short cycles of the Tanner graph. In this section we show that the numbers of short cycles are close to independent Poisson variables, both in a random $s$-LDPC code and conditioned on the code containing a given matrix of type $\tau_n$. A cycle of length $2k$ consists of coordinates $v_1,\dots,v_k$ and checks $c_1,\dots,c_k$ such that $v_i,v_{i+1}\in c_i$ for all $i$, where $v_{k+1} := v_1$ and $c_{k+1} := c_1$. The checks $c_i$ and $c_{i+1}$ belong to different layers. In particular, $k\ge2$. The
\emph{coefficient product} of the cycle is $\hbar := \prod_ir_i\in\F_q^\times$, where $r_i$ is the ratio of the
coefficients of $v_{i+1}$ and $v_i$ in $c_i$. In the reverse direction, the coefficient product of the cycle is $\hbar^{-1}$. Let $\cH\subseteq\F_q^\times$ be a set which contains one element of each set $\inset{\hbar,\hbar^{-1}}$, and for $\hbar\in\cH$ let $Z_{k,\hbar}$ be the number of cycles of length $2k$ with coefficient product $\hbar$ in one of the two directions. In the sums and products over $(k,\hbar)$ below, $k\ge2$ and $\hbar\in\cH$.

We express the means of the $Z_{k,\hbar}$ under $\Epl$ in terms of an operator $Q$. For $r\in\F_q^\times$, let $\pi^{(r)}$ be the joint distribution of two rows $(w_1,w_2)$ of a block whose rows are independent samples of $\nu$, conditioned on $\sum_i\xi_iw_i = 0$ and $\xi_2/\xi_1 = r$. For $\varphi := K\nu$ we have
\[
\pi^{(r)}(a,b) = \frac{\nu(a)\nu(b)\Upsilon(a+rb)}{\sum_z\varphi(z)^s},
\qquad \Upsilon(v) := \sum_z\varphi(z)^{s-2}\chi_z(v).
\]
Let $Q(a,b) := \pi^{(1)}(a,b)/\tau(a)$, the probability that $w_2 = b$ conditioned on $w_1 = a$ when $\xi_2 = \xi_1$. Let $(S_rf)(a) := f(ra)$. Since $\nu$ is scale invariant, we have $\pi^{(r)}(a,b) = \pi^{(1)}(a,rb)$ and
$QS_r = S_rQ$. The operator $Q$ is self-adjoint on $L^2(\tau)$, and $Q\one = \one$.

\begin{lemma}\label{lem:ldpc-cycles}
As $n\to\infty$, the following hold.
\begin{enumerate}[label=\textup{(\arabic*)}]
\item The joint distribution of the $Z_{k,\hbar}$ tends to that of independent Poisson variables with
means $\varpi_{k,\hbar} := \inabs{\inset{\hbar,\hbar^{-1}}}\frac{(s-1)^k\inparen{(t-1)^k+(t-1)(-1)^k}}{2k(q-1)}$.
\item Conditioned on the code containing a given matrix of type $\tau_n$, that is, under $\Epl$, the joint factorial moments of the $Z_{k,\hbar}$ tend to those of independent Poisson variables with means $\varpi_{k,\hbar}(1+\beta_{k,\hbar})$, where $1+\beta_{k,\hbar} := \tr\inparen{Q^kS_\hbar}$ for $\hbar\in\F_q^\times$. The same holds if we also condition on $D_\delta$, that is, with $Y$ in place of $X$. Finally, $\beta_{k,\hbar^{-1}} = \beta_{k,\hbar}$.
\end{enumerate}
\end{lemma}

\subsection{Proof of \texorpdfstring{Theorem~\ref{thm:ldpc-core}}{Theorem 6.1}}\label{sec:ldpc-assembly}

We first write $\E{X^2}/(\E X)^2$ as a sum over the joint types of pairs of matrices of type $\tau_n$. We then bound $\E{Y^2}/(\E Y)^2$ by $C_s$ (Proposition~\ref{prop:ldpc-second}), we show that $C_s = \exp{\sum_{k,\hbar}\varpi_{k,\hbar}\beta_{k,\hbar}^2}$ (Proposition~\ref{prop:ldpc-identity}), and we apply the small subgraph conditioning theorem of Janson~\cite{Jan95}. Let $\omega_0 := \tau\otimes\tau$, the joint type of a pair $(M\,|\,M')$ in which $M$ and $M'$ are independent. For a joint type $\omega$, let $\cM_{n,\omega}$ denote the set of pairs $(M\,|\,M')$ of joint type $\omega$ with $M,M'$ of type $\tau_n$, and let $\cM_{n,\tau_n}$ denote the set of matrices of type $\tau_n$. The contribution of the pairs of joint type $\omega$ to $\E{X^2}/(\E X)^2$ is $\varrho_0(\omega)\inparen{p(\omega)/p(\tau_n)^2}^t$, where $\varrho_0(\omega) := |\cM_{n,\omega}|/|\cM_{n,\tau_n}|^2$ is the probability that two independent uniformly random matrices of type $\tau_n$ have joint type $\omega$. We compute $p(\omega)/p(\tau_n)^2$ using~\eqref{eq:ldpc-bayes} with $\mu = \nu\otimes\nu$ and with $\mu = \nu$, and Lemma~\ref{lem:ldpc-product}.

\begin{lemma}\label{lem:ldpc-product}
A block whose rows are independent samples of $\nu\otimes\nu$ satisfies its parity check with probability $P_{2m}(\nu\otimes\nu) = P_m(\nu)^2$. Conditioned on this event, each row has distribution $\omega_{\nu\otimes\nu} = \omega_0$, and two rows have joint distribution $\pi_{\nu\otimes\nu} = \Eover r{\pi^{(r)}\otimes\pi^{(r)}}$, where $r$ is uniform on $\F_q^\times$. Finally, $\omega_0$ is $\frac{\delta_\tau}2$-smooth, and $\nu\otimes\nu$ is $\frac{\delta_\tau}4$-smooth.
\end{lemma}

\begin{proof}
Let $\mu = \nu\otimes\nu$, and write $w_i = (u_i,u'_i)$ with $u_i,u'_i\in\F_q^m$. The block $w$ satisfies the parity check when $\sum_i\xi_iu_i = 0$ and $\sum_i\xi_iu'_i = 0$. For fixed coefficients, the first check depends only on the $u_i$ and the second check depends only on the $u'_i$, and the $\xi_iu_i$ and the $\xi_iu'_i$ are independent samples of $\nu$ because $\nu$ is scale invariant. Each check holds with probability $P_m(\nu)$ for any coefficients. Hence $P_{2m}(\nu\otimes\nu) = P_m(\nu)^2$, and the distribution of the coefficients, conditioned on both checks holding, is uniform. The $u_i$ and the $u'_i$ are independent conditioned on the coefficients and on both checks holding, and conditioned on the same events, $\xi_1u_1$ has distribution $\omega_\nu = \tau$. Since $\tau$ is scale invariant, $u_1$ also has distribution $\tau$, and $\omega_{\nu\otimes\nu} = \tau\otimes\tau$. Similarly, the joint distribution of $u_1$ and $u_2$ is $\pi^{(r)}$, where $r := \xi_2/\xi_1$ is uniform on $\F_q^\times$, and so $\pi_{\nu\otimes\nu} = \Eover r{\pi^{(r)}\otimes\pi^{(r)}}$.

Note that $\omega_0$ is $\frac{\delta_\tau}2$-smooth. Indeed, let $(y,y')\ne0$, and let $v$ and $v'$ be independent samples of $\tau$. If $y = 0$ or $y' = 0$, then $\ip yv+\ip{y'}{v'}$ is the inner product of a nonzero vector with one sample of $\tau$. Otherwise $\ip{y'}{v'}$ is independent of $v$, and since $\tau$ is scale invariant, for $q\ge3$ it is equal to any given element of $\F_q$ with probability at most $\max\inset{1-\delta_\tau,\frac12}$, and for $q = 2$ it is equal to each element of $\F_2$ with probability at least $\delta_\tau$. For $q\ge3$, $\ip yv+\ip{y'}{v'}$ is nonzero with probability at least $\min\inset{\delta_\tau,\frac12}$, and for $q = 2$ it is zero and nonzero with probability at least $\delta_\tau$ each. In both cases these probabilities are larger than $\frac{\delta_\tau}2$. Finally, $\nu\otimes\nu$ is $\frac{\delta_\tau}4$-smooth. Indeed, by Lemma~\ref{lem:ldpc-smooth}(3), $(\nu\otimes\nu)_a\ge(1-\eta_{\delta_\tau})^2\omega_0(a)\ge\frac12\omega_0(a)$ for all $a$.
\end{proof}

Both marginal distributions of $\omega$ are $\tau_n$. Hence $\prod_a(\nu\otimes\nu)_a^{n\omega_a} = \inparen{\prod_a\nu_a^{n\tau_n(a)}}^2$. By~\eqref{eq:ldpc-bayes} and Lemma~\ref{lem:ldpc-product},
\[
\frac{p(\omega)}{p(\tau_n)^2}\;=\;\frac{\varrho(\omega)}{\varrho_0(\omega)},\qquad
\varrho(\omega) := \frac{\PR{S_{\nu\otimes\nu} = n\omega}}{\PR{S_\nu = n\tau_n}^2},
\]
and the contribution of the pairs of joint type $\omega$ to $\E{X^2}/(\E X)^2$ is $\varrho_0(\omega)^{1-t}\varrho(\omega)^t$. As we will see in Appendix~\ref{app:ldpc-laplace}, if we replace $X$ by $Y$, then the contribution of the pairs whose joint type is not $\delta$-smooth is $o(1)$. The contribution of the other joint types $\omega$ is exponentially small in $n$ if $\omega$ is not close to $\omega_0$. For $\omega$ close to $\omega_0$, both $\varrho_0(\omega)$ and $\varrho(\omega)$ are approximately Gaussian on the space $V$ of vectors in $\R^{\F_q^{2m}}$ whose two marginals are zero, and we compute the sum of these contributions using the local limit theorem. Let $A$ be the restriction of $A_{\nu\otimes\nu}$ (as in Lemma~\ref{lem:ldpc-smooth}(2)) to the functions $y/\omega_0$ with $y\in V$, and let
\[
C_s := \det\nolimits_V(I+A)^{-\frac{t-1}2}\det\nolimits_V\inparen{I-(t-1)A}^{-\frac12}.
\]

\begin{proposition}\label{prop:ldpc-second}
The second moment of $Y$ satisfies $\E{Y^2}\le(C_s+o(1))(\E Y)^2$, where $C_s\to1$ as $s\to\infty$.
\end{proposition}

We prove Proposition~\ref{prop:ldpc-second} in Appendix~\ref{app:ldpc-laplace}.

\begin{fact}[Small subgraph conditioning, {\cite{Jan95}}, see {\cite[Theorems~4.1 and~4.3]{Wor99}}]\label{fact:janson}
Let $\varpi_i>0$ and $\beta_i>-1$ for $i\ge1$, and let $Y_n\ge0$ and
$Z_{i,n}\in\Z_{\ge0}$ be random variables. Suppose that the joint distribution of the $Z_{i,n}$ tends
to that of independent Poisson variables with means $\varpi_i$, that for any finitely
supported $(j_i)$,
$\E{Y_n\prod_i(Z_{i,n})_{j_i}}/\E{Y_n}\to\prod_i\inparen{\varpi_i(1+\beta_i)}^{j_i}$, that
$\sum_i\varpi_i\beta_i^2<\infty$, and that
$\E{Y_n^2}/\E{Y_n}^2\le\exp{\sum_i\varpi_i\beta_i^2}+o(1)$. Then $\PR{Y_n>0}\to1$.
\end{fact}

By Proposition~\ref{prop:ldpc-second}, we may apply Fact~\ref{fact:janson} to $Y$ and the $Z_{k,\hbar}$ if $C_s\le\exp{\sum_{k,\hbar}\varpi_{k,\hbar}\beta_{k,\hbar}^2}$. We show that equality holds.

\begin{proposition}\label{prop:ldpc-identity} For $s$ large enough, $C_s = \exp{\sum_{k,\hbar}\varpi_{k,\hbar}\beta_{k,\hbar}^2}$, where the sum converges, and $|\beta_{k,\hbar}|<1$ for all $k$ and $\hbar$.
\end{proposition}

\begin{proof}
For a multiplicative character $\psi$ of $\F_q$, let $\alpha^\psi_1,\alpha^\psi_2,\dots$ be the eigenvalues of $Q$ on the functions $f$ with $S_rf = \psi(r)f$, except for the eigenvalue $1$ of $\one$. By Lemma~\ref{lem:ldpc-product}, $\pi_{\nu\otimes\nu} = \Eover r{\pi^{(r)}\otimes\pi^{(r)}}$. Thus $A = (s-1)(Q\otimes Q)\Eover r{S_r\otimes S_r}$ on
the functions $y/\omega_0$ with $y\in V$. The operator $\Eover r{S_r\otimes S_r}$ is the projection onto the span of the products of a function in the $\psi$-eigenspace with a function in the $\bar\psi$-eigenspace, and $Q$ has the same eigenvalues on these two eigenspaces, since $Q$ is real and self-adjoint. Thus the eigenvalues of $A$ are $(s-1)\alpha^\psi_i\alpha^\psi_j$ and $0$. On the other hand, $1+\beta_{k,\hbar} = 1+\sum_\psi\psi(\hbar)\sum_i(\alpha^\psi_i)^k$, since $S_\hbar$ acts as $\psi(\hbar)$ on the $\psi$-eigenspace. Since $\beta_{k,\hbar^{-1}} = \beta_{k,\hbar}$, we have $\sum_{\hbar\in\cH}\varpi_{k,\hbar}\beta_{k,\hbar}^2 = \frac{(s-1)^k((t-1)^k+(t-1)(-1)^k)}{2k(q-1)}\sum_{\hbar\in\F_q^\times}\beta_{k,\hbar}^2$, and by Parseval's identity over $\F_q^\times$, $\frac1{q-1}\sum_{\hbar\in\F_q^\times}\beta_{k,\hbar}^2 = \sum_\psi\sum_{i,j}(\alpha^\psi_i\alpha^\psi_j)^k$. For $x := \alpha^\psi_i\alpha^\psi_j$,
\[
\sum_{k\ge1}\frac{(s-1)^k\inparen{(t-1)^k+(t-1)(-1)^k}}{2k}\,x^k
= -\tfrac12\ln\inparen{1-(s-1)(t-1)x}-\tfrac{t-1}2\ln\inparen{1+(s-1)x}.
\]
Combining this with the definition of $C_s$, we obtain the identity. By Lemma~\ref{lem:ldpc-product} and Lemma~\ref{lem:ldpc-smooth}(2) with $\frac{\delta_\tau}2$ in place of $\delta$, the eigenvalues of $A$ are at most $s\eta_{\delta_\tau/2}$ in absolute value, and $(s-1)(\alpha^\psi_i)^2$ is one of them. Thus $|\alpha^\psi_i|^2\le2\eta_{\delta_\tau/2}$ for all $\psi$ and $i$. For $s$ large enough, the series above converge absolutely for every $x = \alpha^\psi_i\alpha^\psi_j$. In particular, $\sum_{k,\hbar}\varpi_{k,\hbar}\beta_{k,\hbar}^2<\infty$. Finally, $|\beta_{k,\hbar}|\le q^m(2\eta_{\delta_\tau/2})^{k/2}<1$ for $k\ge2$.
\end{proof}

In the proof of Proposition~\ref{prop:ldpc-identity}, the identity holds separately for each $k$, and the terms of $\beta_{k,\hbar}$ that depend on the coefficient product $\hbar$ are the terms with $\psi\neq1$. Note that this identity is false in general if we ignore the coefficient products of the cycles.

\begin{proof}[Proof of Theorem~\ref{thm:ldpc-core}]
By Lemma~\ref{lem:ldpc-cycles} and Propositions~\ref{prop:ldpc-second} and~\ref{prop:ldpc-identity}, we may apply the small subgraph conditioning theorem (Fact~\ref{fact:janson}) to $Y$ and the $Z_{k,\hbar}$. Hence, $\PR{Y>0}\to1$, and $Y\le X$ implies $\PR{X>0}\to1$.
\end{proof}

One can check that the proof works with $s_0 = O_{q,\delta_\tau,R_0}\inparen{\log\frac1{\tau_{\min}}+\log\frac1\gamma}$, and, applying this to $\tau = \sigma_m$ and using \cite[Remark~1.10]{MRRSW}, that in Corollary~\ref{cor:ldpc-list} we may take $s = O_{q,\rho}(1/\eps+\log\frac1\gamma)$, with $\gamma$ as in the proof of Corollary~\ref{cor:ldpc-list}.

\section*{Acknowledgements}

We thank Mary Wootters, Anirudh Krishna and Paris Siminelakis for editorial comments on the paper.

\bibliographystyle{alpha}
\bibliography{refs}

\appendix

\section{Deferred proofs}\label{app:deferred}
In this appendix we prove Corollary~\ref{cor:hamming} and Lemmas~\ref{lem:delta}, \ref{lem:activation} and~\ref{lem:tilt-opt}.

\begin{proof}[Proof of Corollary~\ref{cor:hamming}]
Let $M$ be the matrix with columns $x_1,\dots,x_k$, and let $X$ be a uniformly random row of $M$. If $E$ holds, then every column of $M$ and every $\sum_ia_{j,i}x_i$ lies in $B$, and so the coordinates of $X$ and all the $\ip{a_j}{X}$ are $\rho$-sparse. There are at most $q^{nH_q(X)}$ matrices of a given type and $\poly(n)$ types, and $|B|^k = q^{k\,n\,h_q(\rho)-O(\log n)}$. Thus $\PR E\le\poly(n)\,q^{n(H^*-k\,h_q(\rho))}$, where $H^*$ is the largest value of $H_q(X)$ over the types of $M$ for which $E$ holds. The two bounds follow from parts~(1) and~(2) of Theorem~\ref{thm:tool}.
\end{proof}

\begin{proof}[Proof of Lemma~\ref{lem:delta}] For $U := u+\F_q\onesvector$ and any $x\in U$, let $f_U$ be the distribution on $\F_q$ given by $f_U(c):= \mu_q^m(x+c\onesvector)/\mu_q^m(U)$. Then $\Delta_m = H_q\inparen{u\mid u\bmod\F_q\onesvector}$ is the expectation of $H_q(f_U)$. Shannon entropy is at most R\'enyi entropy of order $\frac12$, which implies that $H_q(f_U)\le2\log_q\sum_c\sqrt{f_U(c)}$. By Jensen's inequality,
\[
\Delta_m\;\le\;2\log_q\sum_U\mu_q^m(U)\sum_c\sqrt{f_U(c)}\;=\;2\log_q\sum_U\sum_c\sqrt{\mu_q^m(x+c\onesvector)\,\mu_q^m(U)}.
\]
The sum on the right-hand side is at most $\sum_{d\in\F_q}\sum_{y\in\F_q^m}\sqrt{\mu_q^m(y)\,\mu_q^m(y+d\onesvector)}$, since $\sqrt{\mu_q^m(U)}\le\sum_{c'}\sqrt{\mu_q^m(x+c'\onesvector)}$. Since $\mu_q(a) = (1-\rho)z_0^{\one\inbrak{a\neq0}}$ for $a\in\F_q$, for $d\neq0$ the sum over $y$ is a product over the coordinates, and by Lemma~\ref{lem:lambda} each factor is $(1-\rho)S(d) = \lambda$. Thus the sum is at most $1+(q-1)\lambda^m$, and so $\Delta_m\le2\log_q\inparen{1+(q-1)\lambda^m}\le\frac{2(q-1)}{\ln q}\lambda^m$.
\end{proof}

\begin{proof}[Proof of Lemma~\ref{lem:activation}]
For a coordinate $i$ with $a_i\neq0$, as in the proof of \cite[Claim~4.4]{MRRSW}, we have $\Eover{\mu_q^k}{\omega_p^{-\mathrm{Tr}(a_iY_i)}} = (1-\rho)+\rho\cdot\frac{-1}{q-1} = \kappa$, and for each coordinate $i$ with $a_i = 0$ this expectation is $1$. Since $\mu_q^k$ is a product distribution, $\widehat{\mu_q^k}(a) = q^{-k}\prod_i\Eover{\mu_q^k}{\omega_p^{-\mathrm{Tr}(a_iY_i)}} = q^{-k}\kappa^w$. Similarly, $\E{\chi_c(\ip aY)} = \kappa^w$ for $c\in\F_q^\times$, since $ca$ has support of size $w$. Since $\one\inbrak{y = 0} = \frac1q\sum_{c\in\F_q}\chi_c(y)$ for $y\in\F_q$, we have $\PROver{\mu_q^k}{\ip aY = 0} = \frac1q\inparen{1+(q-1)\kappa^w}$, and so $\PROver{\mu_q^k}{\ip aY\neq0} = \frac{q-1}q(1-\kappa^w) = \pi_w$. Note that $\mu_q^k$ is invariant under multiplication by elements of $\F_q^\times$. Hence the distribution of $\ip aY$, conditioned on $\ip aY\neq0$, is uniform on $\F_q^\times$, which establishes \eqref{eq:activation-identity}. The claims about $\pi_w$ follow from the definition of $\kappa\in(0,1)$.
\end{proof}

\begin{proof}[Proof of Lemma~\ref{lem:tilt-opt}]
By \eqref{eq:chernoff} with $z_t = z^*$ for every $t$, by the assumption on $M$, and since $1-\pi_2(1-z^*) = \frac{1-\pi_2}{1-\rho}$,
\[
H_q(X)\;\le\;k\,h_q(\rho)+T\inparen{\rho\log_q\frac1{z^*}+\log_q\frac{1-\pi_2}{1-\rho}}+\log_q M,
\]
and the coefficient of $T$ is $\rho\log_q\frac{\pi_2}{\rho}+(1-\rho)\log_q\frac{1-\pi_2}{1-\rho} = -\gamma^*$.
\end{proof}

\section{Proofs for random LDPC codes}\label{app:ldpc}
In this appendix we prove Lemmas~\ref{lem:ldpc-layer}, \ref{lem:ldpc-smooth}, \ref{lem:ldpc-planted} and~\ref{lem:ldpc-cycles}
and Proposition~\ref{prop:ldpc-second}. In the proofs of Lemmas~\ref{lem:ldpc-one} and~\ref{lem:ldpc-cycles} and Proposition~\ref{prop:ldpc-second}, we compute probabilities of the form $\PR{S_\mu = n\omega}$ using the following local limit theorem.

\begin{fact}[Local limit theorem, {\cite[Theorem~22.1]{BR76}}]\label{fact:llt} Let $\inset{\mathcal D_\theta}$ be a family of distributions on $\Z^\ell$ that depend continuously on $\theta$ in a compact set and have the same finite support, and let $\Gamma$ be the lattice generated by the differences of the points of this support, with determinant $\det(\Gamma)$. Let $S_N$ be a sum of $N$ independent samples of $\mathcal D_\theta$, with expectation $N\bar\mu_\theta$ and covariance $N\Sigma_\theta$ of full rank on $\Gamma\otimes\R$. Then, uniformly in $\theta$ and in $x$ in
the coset of $S_N$ modulo $\Gamma$,
\[
\PR{S_N = x} = \frac{\det(\Gamma)}{(2\pi N)^{r/2}\det(\Sigma_\theta)^{1/2}}
\inparen{e^{-\frac1{2N}(x-N\bar\mu_\theta)^{\top}\Sigma_\theta^{-1}(x-N\bar\mu_\theta)}+o(1)},\qquad r := \mathrm{rank}\,\Gamma.
\]
\end{fact}

The proof in \cite{BR76} is uniform over such a family. Indeed, the Fourier transforms depend continuously on $\theta$, and, for any $\eps>0$, their absolute values are bounded away from $1$ at the points at distance at least $\eps$ from the dual lattice, uniformly in $\theta$.

We now determine the lattice $\Gamma$ for the sums $S_\mu$ of Section~\ref{sec:ldpc-blocks}, where $\supp(\mu) = \F_q^d$. Each $S_\mu$ is a sum of $n/s$ independent samples of $s$ times the type of $w$, conditioned on $\sum_i\xi_iw_i = 0$. For $q\ge3$, there are $\xi_1,\xi_2,\xi_3\in\F_q^\times$ that sum to $0$, and so the blocks $(a,a,a,0,\dots,0)$ and $(a,a,0,\dots,0)$ satisfy the parity check. The difference of their types is $(e_a-e_0)/s$. Hence $\Gamma$ is the lattice of integer vectors whose entries sum to $0$. For $q = 2$, a block satisfies the parity check if and only if $\sum_iw_i = 0$, and the lattice generated by $s$ times the differences of the types of the blocks $(a,a,0,\dots,0)$ and $(a,b,a+b,0,\dots,0)$ is the lattice of integer vectors $N$ with $\sum_aN_a = 0$ and $\sum_aN_a\,a = 0$ in $\F_2^d$. Thus, $\Gamma$ has rank $q^d-1$, and its index $\iota_d$ in the lattice of integer vectors whose entries sum to $0$ is $1$ for $q\ge3$ and $2^d$ for $q = 2$. For $q = 2$, $n\omega$ belongs to the coset of $S_\mu$ modulo $\Gamma$ if every column of a matrix of type $\omega$ has even weight. The columns of every matrix of type $\tau_n$, and of every pair $(M\,|\,M')$ with $M,M'$ of type $\tau_n$, have even weight.

\subsection{Proof of \texorpdfstring{Lemmas~\ref{lem:ldpc-layer} and~\ref{lem:ldpc-smooth}}{Lemmas 6.2 and 6.3}}\label{app:ldpc-layer} We first prove Lemma~\ref{lem:ldpc-layer}. Let $\Lambda$ be as in Section~\ref{sec:ldpc-blocks}, and assume that $\Lambda$ is independent of the code. The distribution of $\Lambda$, conditioned on $\Lambda$ having type $\omega$, is uniform on the set of matrices of type $\omega$. The probability that $\Lambda$ satisfies the parity checks of one layer, conditioned on this event, is $p(\omega)$. On the other hand, $\Lambda$ satisfies these parity checks with probability $P_d(\mu)^{n/s}$, and the distribution of the blocks of $\Lambda$, conditioned on this event, is identical to the distribution of $n/s$ independent samples of $w$, each conditioned on $\sum_i\xi_iw_i = 0$. Thus, the probability that $\Lambda$ satisfies the parity checks and has type $\omega$ is equal to both $\binom n{n\omega}\prod_a\mu_a^{n\omega_a}\,p(\omega)$ and $P_d(\mu)^{n/s}\,\PR{S_\mu = n\omega}$, and so we obtain~\eqref{eq:ldpc-bayes}.

We now prove Lemma~\ref{lem:ldpc-smooth}. For~(1), we compute $\omega_\mu$ and $\pi_\mu$. Using the decomposition of the indicator of the event that $\sum_i\xi_iw_i = 0$ in Section~\ref{sec:sparsity}, we can express $\omega_\mu$ and $\pi_\mu$ as
\[
\omega_\mu(a) = \mu_a\,\frac{\sum_z(K\one_a)(z)\,((K\mu)(z))^{s-1}}{\sum_z((K\mu)(z))^s},\qquad
\pi_\mu(a,b) = \mu_a\mu_b\,\frac{\sum_z(K\one_a)(z)\,(K\one_b)(z)\,((K\mu)(z))^{s-2}}{\sum_z((K\mu)(z))^s}.
\]
In each sum the $z = 0$ term is $1$. Since $\mu$ is $\frac\delta2$-smooth and $|K\one_a|\le1$, the sum of the absolute values of the other terms is at most $q^d\kappa_{\delta/2}^{\,s-2}\le\frac{\eta_\delta}8$. Item~(1) follows from $\frac{\eta_\delta}8\le\frac1{64}$.

For~(2), note that $s$ times the type of $w$ is the sum of the indicators of its $s$ rows. The covariance of this sum is the sum of $s$ times $\Sigma(\omega_\mu)$, the covariance of one indicator, and $s(s-1)$ times $E := \pi_\mu-\omega_\mu\omega_\mu^{\top}$, the covariance of the indicators of two different rows, that is, $s\Sigma(\omega_\mu)(I+A_\mu)$. Note that $\Sigma(\omega_\mu)(I+A_\mu) = \Sigma(\omega_\mu)+(s-1)E$, and that $\Sigma(\omega_\mu)$ and $E$ vanish on constant vectors, since the rows of $\pi_\mu$ sum to $\omega_\mu$. Hence it suffices to show that $(s-1)|y^{\top}Ey|\le s\eta_\delta\,y^{\top}\Sigma(\omega_\mu)y$ for every $y\in\R^{\F_q^d}$ with $\sum_a\omega_\mu(a)y_a = 0$. For such $y$ we have $y^{\top}\Sigma(\omega_\mu)y = \sum_a\omega_\mu(a)y_a^2$. By item~(1) and the Cauchy--Schwarz inequality, $|y^{\top}Ey|\le\eta_\delta\inparen{\sum_a\mu_a|y_a|}^2\le\eta_\delta(1+\eta_\delta)\sum_a\omega_\mu(a)y_a^2$. It follows that $(s-1)|y^{\top}Ey|\le s\eta_\delta\,y^{\top}\Sigma(\omega_\mu)y$.

For~(3), we find $\mu^*$ as a fixed point of a map $\Xi$. Let $\cR$ be the set of scale-invariant $\frac\delta2$-smooth distributions, and note that $\cR$ is compact and convex. For $\mu\in\cR$, let $\Xi(\mu)$ be the distribution proportional to $\inparen{\omega_a\mu_a/\omega_\mu(a)}_a$. By item~(1), $|\Xi(\mu)_a/\omega_a-1|\le\frac32\eta_\delta$. Hence $\Xi(\mu)$ is
$\frac\delta2$-smooth, since $\omega$ is $\delta$-smooth and $\eta_\delta\le\frac18$. Note that $\Xi(\mu)$ is scale invariant, because $\omega$, $\mu$ and $\omega_\mu$ are scale invariant. By Brouwer's fixed point
theorem, the continuous map $\Xi$ has a fixed point $\mu^*\in\cR$. Then $\omega_{\mu^*}$ is proportional to $\omega$, and hence equal to $\omega$. By item~(1), $|\mu^*_a/\omega_a-1|\le\eta_\delta$ for all $a$.

\subsection{Proof of \texorpdfstring{Lemma~\ref{lem:ldpc-planted}}{Lemma 6.5}}\label{app:ldpc-planted}

We may assume that $\delta<\frac{\delta_\tau}2$, that $\kappa' := \bar\kappa+2\delta<1$ and that $h_q(\delta)\le\frac{1-R_0}2$. Fix $M$ of type $\tau_n$. By symmetry $\Epl[\one_{D_\delta^c}] = \PR{D_\delta^c\mid M\subseteq\cC}$. Let
\[
N_k := \sum_x\frac{\PR{(M\,|\,x)\subseteq\cC}}{\PR{M\subseteq\cC}},
\]
where the sum is over the $x\in\F_q^n$ of weight $k$. By the union bound, it suffices to show that the $N_k$ sum to $o(1)$ over $1\le k\le\delta n$. Let $S := \supp x$. For small $k$ we use the fact that $x$ does not satisfy the parity check of a block that contains exactly one coordinate of $S$ (Claim~\ref{claim:ldpc-small}), and for larger $k$ we use~\eqref{eq:ldpc-bayes} as in the proof of \cite[Lemma~2.13]{MRRSW} (Claim~\ref{claim:ldpc-medium}).

Since $t = (1-R)s\ge(1-R_0)s$, we have $t\ge3$ for $s$ large enough. Let $c_\tau := e/(2\tau_{\min}^2)$, and let $\beta_0$ be the largest number such that $eq(c_\tau s)^{t/2}\beta_0^{t/2-1}\le\frac12$.

\begin{claim}\label{claim:ldpc-small}
The $N_k$ with $k\le\beta_0n$ sum to $o(1)$.
\end{claim}

\begin{proof}
Let $p_b(M)$ denote the probability that $M$ satisfies the parity check of block $b$, where the probability is over the choice of the coefficients. The vector $x$ does not satisfy the parity check of a block that contains exactly one coordinate of $S$. Hence $(M\,|\,x)$ satisfies the parity checks of a layer with probability at most $\Eover\sigma{\prod_bp_b(M)\one[\text{no block meets $S$ once}]}$, where $\sigma$ is the permutation of the layer. The ratio of this expectation to $\Eover\sigma{\prod_bp_b(M)}$ is the probability that no block contains exactly one coordinate of $S$, conditioned on the event that $M$ satisfies the parity checks of the layer.
The distribution of $\sigma$, conditioned on this event, depends only on the collection of rows of $M$ in each block. Conditioned also on these collections, $\sigma$ maps the $\tau_n(a)n$ coordinates with row $a$ to the positions with row $a$ by a uniformly random bijection.

If no block contains exactly one coordinate of $S$, then $S$ is contained in the union of some $j\le k/2$ blocks, and for a fixed set of $j$ blocks, the probability that $S$ is contained in their union is at most $(js/(\tau_{\min}n))^k$. By the union bound over the sets of blocks, the probability that $(M\,|\,x)$ satisfies the parity checks of the layer, conditioned on the event that $M$ satisfies these parity checks, is at most
\[
\sum_{j\le k/2}\binom{n/s}j\inparen{\frac{js}{\tau_{\min}n}}^k\;\le\;k\inparen{\frac{c_\tau sk}n}^{k/2}.
\]
The layers are independent, and they are also independent conditioned on $M\subseteq\cC$. There are at most $(enq/k)^k$ vectors of
weight $k$. Thus
\[
N_k\;\le\;k^t\inbrak{eq\,(c_\tau s)^{t/2}(k/n)^{t/2-1}}^k .
\]
For $k\le\log^2n$, the right-hand side is at most $n^{-(t/2-1)}\cdot\poly(\log n)$. For
larger $k\le\beta_0n$ we have $eq\,(c_\tau s)^{t/2}(k/n)^{t/2-1}\le\frac12$. Thus $N_k\le k^t2^{-k}$. Hence the $N_k$ with $k\le\beta_0n$ sum to $o(1)$.
\end{proof}

\begin{claim}\label{claim:ldpc-medium}
The $N_k$ with $\beta_0n\le k\le\delta n$ sum to $o(1)$.
\end{claim}

\begin{proof}
Let $\beta_0n\le k\le\delta n$ and $\beta := k/n$, and let $\omega'$ be the type of $(M\,|\,x)$. By~\eqref{eq:ldpc-bayes} with $\mu = \omega'$, since $\PR{S_{\omega'} = n\omega'}\le1$ and $\binom n{n\omega'}\prod_a\omega'(a)^{n\omega'(a)}\ge(n+1)^{-q^{m+1}}$, we have $p(\omega')\le(n+1)^{q^{m+1}}P_{m+1}(\omega')^{n/s}$, as in the proof of \cite[Lemma~2.13]{MRRSW}. Let $u_\beta := \inparen{1-\frac{q\beta}{q-1}}^s$. For $y'\in\F_q$, $(K\omega')(0,y')$ is $1$ if $y' = 0$ and $1-\frac{q\beta}{q-1}$ otherwise. Hence $(K\omega')(0,y')^s = u_\beta$ for $y'\neq0$. For $y\ne0$, $|(K\omega')(y,y')|\le\bar\kappa+2\delta\le\kappa'$. Indeed, the contributions of the rows with $x = 0$ to the expectations that define $(K\omega')(y,y')$ and $(K\tau_n)(y)$ are equal, and a fraction $\beta\le\delta$ of the rows have $x\neq0$. Since $1+(q-1)u_\beta\ge1$, we obtain
\[
P_{m+1}(\omega')\;\le\;q^{-m}\cdot\frac{1+(q-1)u_\beta}q\inparen{1+q^{m+1}\kappa'^s}.
\]
There are at most $q^{nh_q(\beta)}$ vectors of weight $k$, and by Lemma~\ref{lem:ldpc-one},
\[
p(\tau_n)^t\;\ge\;q^{-(1-R)mn-O(n(1-R)q^m\bar\kappa^{\,s})-O(\log n)}.
\]
Hence
\[
N_k\;\le\;e^{O(n(1-R)q^{m+1}\kappa'^s)+O(t\log n)}\,q^{n\inbrak{h_q(\beta)-(1-R)\inparen{1-\log_q(1+(q-1)u_\beta)}}}.
\]
Note that $u_\beta\le e^{-s\beta}$. If $s\beta$ is larger than a constant that depends only on $q$, then $\log_q(1+(q-1)u_\beta)\le\frac14$, and since $h_q(\beta)\le h_q(\delta)\le\frac{1-R}2$, the exponent is at most $-\frac{1-R}4n$. Otherwise, by the concavity of the logarithm and of $1-e^{-x}$, we have $1-\log_q(1+(q-1)u_\beta)\ge\frac{(q-1)(1-u_\beta)}{q\ln q}\ge c_qs\beta$ for some $c_q>0$. Note that $h_q(\beta)\le\beta\log_q\frac{e(q-1)}{\beta_0}$ and $\log\frac1{\beta_0} = O\inparen{\log\frac{qs}{\tau_{\min}}}$. For $s$ large enough the exponent is at most $-\frac12(1-R)c_qsn\beta$. In both cases the exponent is at most $-n$ times a polynomial in $1/s$ and $\tau_{\min}$, while $q^{m+1}\kappa'^s$ is exponentially small in $s$. Thus $N_k\le e^{-\Omega(n)}$, and the sum is $o(1)$.
\end{proof}

\subsection{Proof of \texorpdfstring{Lemma~\ref{lem:ldpc-cycles}}{Lemma 6.6}}\label{app:ldpc-cycles}

For~(1), we proceed by the method of moments. As in Section~\ref{sec:ldpc-cycles}, a cycle of length $2k$ consists of coordinates $v_1,\dots,v_k$ and checks $c_1,\dots,c_k$, where $c_i$ and $c_{i+1}$ belong to different layers for all $i$, and the indices are taken modulo $k$. For each layer, the pairs $v_i,v_{i+1}$ of coordinates of the cycle that share a block of this layer are disjoint. In each layer, the probability that the two coordinates of each of $j$ disjoint pairs of coordinates lie in the same block, and different pairs lie in different blocks, is $\inparen{\frac{s-1}{n-1}}^j(1+O(\frac1n))$, and the ratios of the coefficients of these pairs are independent and uniformly random in $\F_q^\times$. There are
$\tr(J_t-I)^k = (t-1)^k+(t-1)(-1)^k$ sequences of layers, where $J_t$ is the all-ones matrix. Thus $\E{Z_{k,\hbar}} = \frac{n^k(1+O(1/n))}{2k}\tr(J_t-I)^k\inparen{\frac{s-1}n}^k\frac{|\inset{\hbar,\hbar^{-1}}|}{q-1}\to\varpi_{k,\hbar}$, since for a fixed order of the coordinates of the cycle the coefficient product is $\hbar$ or $\hbar^{-1}$ with probability $|\inset{\hbar,\hbar^{-1}}|/(q-1)$. 

In the factorial moments, the sum of the terms of collections of cycles that share a coordinate or a check is $O(\frac1n)$. Indeed, each shared coordinate or check reduces the number of such collections by a factor of order $n$, and increases their probability by at most a constant factor. The sum of the terms of collections of disjoint cycles is the product of the means. For $t\ge3$, every $\varpi_{k,\hbar}$ with $k\ge2$ is positive.

For~(2), we compute the factorial moments under $\Epl$ as in~(1), with the probabilities in each layer conditioned on $M\subseteq\cC$. Fix a layer, $j$ disjoint pairs of coordinates with rows $(a_\ell,a'_\ell)$ of $M$, and ratios $r_\ell$, and let $E$ be the event that the two coordinates of each pair lie in the same block, different pairs lie in different blocks, and the coefficients of the $\ell$-th pair have ratio $r_\ell$. Let $w^{(1)},\dots,w^{(n/s)}$ be independent samples of $w$ for $\mu = \nu$, each conditioned on $\sum_i\xi_iw_i = 0$. As in the proof of Lemma~\ref{lem:ldpc-layer}, conditioned on $M$ satisfying the parity checks of the layer, the blocks of $M$ have the distribution of $w^{(1)},\dots,w^{(n/s)}$, conditioned on the matrix with these blocks having type $\tau_n$. The coordinates with row $a$ are mapped to the positions with row $a$ by a uniformly random bijection. Conditioned on two given rows and a given ratio of their coefficients in each of $j$ blocks, the probability that the matrix with these blocks has type $\tau_n$ is $1+o(1)$ times the probability without this conditioning, by Fact~\ref{fact:llt}. Indeed, the conditioning affects only $j$ of the $n/s$ blocks, and the difference between the means of the sum with and without this conditioning is $O(1)$.

In each block, the probability that two given positions have rows $(a,a')$ and coefficients with ratio $r$ is $\pi^{(r)}(a,a')/(q-1)$. The probability that the two coordinates of the $\ell$-th pair are mapped to two given positions with rows $(a_\ell,a'_\ell)$ is $(n^2\tau(a_\ell)\tau(a'_\ell))^{-1}(1+o(1))$. Taking the sum over the $j$ blocks and over two positions in each of them, and using $\PR E = \inparen{\frac{s-1}{n-1}}^j(q-1)^{-j}(1+O(\frac1n))$, we obtain
\[
\frac{\PR{E\mid\text{$M$ satisfies the parity checks of the layer}}}{\PR E}\;=\;\prod_{\ell=1}^j\frac{\pi^{(r_\ell)}(a_\ell,a'_\ell)}{\tau(a_\ell)\tau(a'_\ell)}\,(1+o(1)).
\]

For a single cycle, let $a_i$ denote the row of $M$ at $v_i$, and note that the rows of $k$ different coordinates are asymptotically independent samples of $\tau$. Let $b_1 := a_1$ and $b_{i+1} := r_1\cdots r_ia_{i+1}$. The mean of $Z_{k,\hbar}$ under $\Epl$ is $\varpi_{k,\hbar}$ times
\[
\sum_{a_1,\dots,a_k}\prod_i\tau(a_i)\prod_i\frac{\pi^{(r_i)}(a_i,a_{i+1})}{\tau(a_i)\tau(a_{i+1})}
= \sum_b\prod_iQ(b_i,b_{i+1})\Big|_{b_{k+1} = \hbar b_1} = \tr\inparen{Q^kS_{\hbar^{-1}}}.
\]
Here we used
$\pi^{(r)}(a,b) = \pi^{(1)}(a,rb)$ and $\pi^{(1)}(ca,cb) = \pi^{(1)}(a,b)$ and $\tau(ca) = \tau(a)$ for $c\in\F_q^\times$. For the reverse order of the coordinates of the cycle, the coefficient product is $\hbar^{-1}$, and since $\pi^{(r^{-1})}(b,a) = \pi^{(r)}(a,b)$, the sum above is the same. Hence $\tr\inparen{Q^kS_{\hbar^{-1}}} = \tr\inparen{Q^kS_\hbar}$.

As in~(1), the sum of the terms of collections of disjoint cycles is the product of the means. When cycles share a coordinate or a check, a block may contain more than two of the fixed rows. For such a block, in place of $\pi^{(r)}(a,a')/(\tau(a)\tau(a'))$ we get the ratio of the probability that $w$ has these rows in the fixed positions, for $\mu = \nu$ and conditioned on $\sum_i\xi_iw_i = 0$, to the product of their probabilities under $\tau$. This ratio is at most $\tau_{\min}^{-c}$, where $c$ is the number of fixed rows in the block. The sum of the terms of collections of cycles that share a coordinate or a check is again $O(\frac1n)$, as in~(1). Finally, the difference between these moments for $X$ and for $Y$ is $o(1)$, by Cauchy--Schwarz, Lemma~\ref{lem:ldpc-planted}, and the fact that the moments of the numbers of cycles under $\Epl$ are bounded.

\subsection{Proof of \texorpdfstring{Proposition~\ref{prop:ldpc-second}}{Proposition 6.8}}\label{app:ldpc-laplace}

We write $\E{Y^2}$ as a sum over the joint types of the pairs $(M\,|\,M')$. In Claim~\ref{claim:ldpc-lowrank} we show that the terms of the pairs whose joint type is not $\delta$-smooth sum to $o((\E X)^2)$. In Claim~\ref{claim:ldpc-far} we bound each of the other terms, and in Claim~\ref{claim:ldpc-near} we use these bounds and Fact~\ref{fact:llt} to show that they sum to $(C_s+o(1))(\E X)^2$. Let $\Omega_n$ be the set of joint types $\omega$ on $\F_q^{2m}$ of pairs $(M\,|\,M')$ with $M,M'$
of type $\tau_n$, let $\omega_n := \tau_n\otimes\tau_n$, and let $I_n(\omega) := D(\omega\,\|\,\omega_n) = 2H(\tau_n)-H(\omega)$ and $I(\omega) := D(\omega\,\|\,\omega_0)$. Recall that $\E X = |\cM_{n,\tau_n}|\,p(\tau_n)^t$.

We split the pairs according to their rank and whether their joint type is $\delta$-smooth. A pair $(M\,|\,M')$ of rank $d$ has a set $B$ of $d$ of its columns that form a basis for the column span of $(M\,|\,M')$. If $D_\delta$ holds and $(M\,|\,M')\subseteq\cC$, then $B\subseteq\cC$, and the type $\omega_B$ of $B$ satisfies $\PROver{v\sim\omega_B}{\ip zv\ne0}\ge\delta$ for all $z\ne0$, by \cite[Remark~2.12]{MRRSW}. In particular, for $q\ge3$, if $D_\delta$ holds and a pair of rank $2m$ is contained in $\cC$, then the joint type of the pair is $\delta$-smooth. Recall that for $q = 2$ we also require that $\PR{\ip zv=0}\ge\delta$ in the definition of $\delta$-smoothness. Hence the second sum below also contains pairs of rank $2m$, and
\begin{equation}\label{eq:ldpc-split}
\E{Y^2}\;\le\;\sum_{\substack{\omega\in\Omega_n\\ \omega\ \delta\text{-smooth}}}|\cM_{n,\omega}|\,p(\omega)^t
\;+\;\sum_{d\le2m}\ \sum_{\tau'}\ |\cM_{n,\tau'}|\ \max_{B}\PR{B\subseteq\cC},
\end{equation}
where the inner sum is over the joint types $\tau'$ of pairs of rank $d$ such that $\tau'$ is not $\delta$-smooth if $d = 2m$, and the maximum is over the sets $B$ of $d$ columns of such a pair that form a basis for the column span of the pair and satisfy $\PROver{v\sim\omega_B}{\ip zv\ne0}\ge\delta$ for all $z\ne0$.

\begin{fact}[{\cite[Claim~3.1]{MRRSW}}]\label{fact:ldpc-pairs}
If $R\ge\Rs\tau+\gamma$, and $\tau'$ is the joint type of a pair
$(M\,|\,M')$ with $M,M'$ of type $\tau$ and of rank $d<2m$, then
$H_q(\tau')-(1-R)d\le2\inparen{H_q(\tau)-(1-R)m}-\gamma$.
\end{fact}

\begin{fact}[{\cite[Lemma~2.13]{MRRSW}}]\label{fact:ldpc-smooth}
Let $B\in\F_q^{n\times d}$ have type $\omega$ with $\PROver{v\sim\omega}{\ip zv\ne0}\ge\delta$ for all $z\ne0$, and let $\eps'>0$. Then for any odd $s\ge s_0(\eps',\delta,q,d)$ and $n$ large enough,
\[
\PR{B\subseteq\cC}\;\le\;q^{-(1-\eps')(1-R)dn}.
\]
\end{fact}

\begin{claim}\label{claim:ldpc-lowrank}
The second sum in~\eqref{eq:ldpc-split} is $o((\E X)^2)$.
\end{claim}

\begin{proof}
We first bound the terms with $d<2m$. By
Fact~\ref{fact:ldpc-smooth} with $\eps' = \frac\gamma{8m}$, $\PR{B\subseteq\cC}\le q^{-(1-R)dn+\frac\gamma4n}$, since $d<2m$. Note that $\Rs{\tau_n}\to\Rs\tau$ (see the proof of Lemma~\ref{fact:threshold}). Hence for $n$ large enough we have $R\ge\Rs{\tau_n}+\frac\gamma2$, and by Fact~\ref{fact:ldpc-pairs} applied to $\tau_n$, $H_q(\tau')-(1-R)d\le2\inparen{H_q(\tau_n)-(1-R)m}-\frac\gamma2$. By Lemma~\ref{lem:ldpc-one}, $\E X\ge q^{n(H_q(\tau_n)-(1-R)m)-O(n(1-R)q^m\bar\kappa^s)-O(\log n)}$. Since $|\cM_{n,\tau'}|\le q^{nH_q(\tau')}$, each term with $d<2m$ is at most
\[
q^{n(H_q(\tau')-(1-R)d)+\frac\gamma4n}\;\le\;q^{2n(H_q(\tau_n)-(1-R)m)-\frac\gamma4n}\;\le\;q^{-(\frac\gamma4-O(q^m\bar\kappa^{\,s}))n+O(\log n)}(\E X)^2.
\]
Since there are $\poly(n)$ joint types, these terms sum to $o((\E X)^2)$ for $s$ large enough.

It remains to bound the terms with $q = 2$ and $d = 2m$. If $\tau'$ is not $\delta$-smooth, then $\PROver{v\sim\tau'}{\ip zv = 0}<\delta$ for some $z = (y,y')\ne0$. Since the marginal distributions of $\tau'$ are $\tau_n$, we have $y\ne0$ and $y'\ne0$, and so, as in the proof of Lemma~\ref{lem:ldpc-product}, $\PROver{v\sim\omega_n}{\ip zv = 0}\ge\delta_\tau-o(1)$. Since $\delta\le\frac{\delta_\tau}2$, it follows that $|\tau'-\omega_n|_1\ge\delta_\tau-o(1)$, and by Pinsker's inequality $I_n(\tau')\ge\frac{\delta_\tau^2}4$ for $n$ large enough. Hence, $|\cM_{n,\tau'}|\le e^{2nH(\tau_n)-n\delta_\tau^2/4}$. Applying Fact~\ref{fact:ldpc-smooth} with $d = 2m$ and $\eps' = \frac{\delta_\tau^2}{16m\ln2}$, we get $\PR{B\subseteq\cC}\le2^{-2(1-R)mn}e^{n\delta_\tau^2/8}$. Since there are $\poly(n)$ joint types, these terms sum to at most $e^{-n\delta_\tau^2/8+O(n(1-R)2^m\bar\kappa^s)+O(\log n)}(\E X)^2$, which is $o((\E X)^2)$ for $s$ large enough.
\end{proof}

For the smooth pairs, recall from Section~\ref{sec:ldpc-assembly} that the term in the first sum of~\eqref{eq:ldpc-split} that corresponds to $\omega$ is $\varrho_0(\omega)^{1-t}\varrho(\omega)^t(\E X)^2$.

For a function $g$ on $\F_q^{2m}$, let $\chi^2(g) := \sum_a\omega_0(a)g(a)^2$, and for a vector $\Delta$ let $\chi^2(\Delta) := \chi^2(\Delta/\omega_0)$. Let $W$ be the space of functions on $\F_q^m$ with mean $0$ under $\tau$. A vector $y$ on $\F_q^{2m}$ corresponds to the function $y/\omega_0$, and $V$ corresponds to $W\otimes W$. Let $d_V := \dim V = (q^m-1)^2$.

\begin{claim}\label{claim:ldpc-far} Let $\omega\in\Omega_n$ be $\delta$-smooth. Then
$\varrho_0(\omega)^{1-t}\varrho(\omega)^{t}\le e^{-nI_n(\omega)/2+O(\log n)}$. If additionally $|\omega_a/\omega_0(a)-1|\le\frac12$ for all $a$, then $\varrho_0(\omega)^{1-t}\varrho(\omega)^{t}\le Cn^{-d_V/2}e^{-nI_n(\omega)/2}$ for some constant $C$.
\end{claim}

\begin{proof} We first write $\varrho(\omega)/\varrho_0(\omega)$ using~\eqref{eq:ldpc-bayes}, we then bound its exponent using Taylor's theorem, and finally we bound $\varrho_0(\omega)$ and the two probabilities in~\eqref{eq:ldpc-bayes}. We define $P_{2m}(x)$ for any $x\in\R^{\F_q^{2m}}$ by the sum that defines $P_{2m}(\mu)$ in Section~\ref{sec:ldpc-blocks}. Let $g := \omega/\omega_0-1$, and note that $g\ge-1$ and $\sum_a\omega_0(a)g(a) = 0$. For $u\in[0,1]$, let $\mu_u$ be the distribution proportional to $(\nu\otimes\nu)(1+ug)$, and let $F(u) := \ln P_{2m}\inparen{(\nu\otimes\nu)(1+ug)}$. Note that $\prod_a\inparen{(\nu\otimes\nu)_a(1+g_a)}^{n\omega_a} = \inparen{\prod_a\nu_a^{n\tau_n(a)}}^2e^{nI(\omega)}$ and that $F(0) = 2\ln P_m(\nu)$ by Lemma~\ref{lem:ldpc-product}. By~\eqref{eq:ldpc-bayes} with $\mu = \mu_1$ and with $\mu = \nu$,
\[
\frac{\varrho(\omega)}{\varrho_0(\omega)}\;=\;\frac{p(\omega)}{p(\tau_n)^2}\;=\;\frac{\PR{S_{\mu_1} = n\omega}}{\varrho_0(\omega)\,\PR{S_\nu = n\tau_n}^2}\,e^{-nI(\omega)+\frac ns(F(1)-F(0))}.
\]
Here we replaced $\mu_1$ by $(\nu\otimes\nu)(1+g)$, since $P_{2m}(cx) = c^sP_{2m}(x)$ for $c>0$.

Next we bound $F(1)-F(0)$. Let $g_u := g/(1+ug)$, and let $\EE_u$ and $\mathrm{Cov}_u$ denote the expectation and the covariance over $(w,\xi)$ for $\mu = \mu_u$, conditioned on $\sum_i\xi_iw_i = 0$. Then $F'(u) = \EE_u\sum_ig_u(w_i)$, and $F''(u) = \mathrm{Var}_u\inparen{\sum_ig_u(w_i)}-\EE_u\sum_ig_u(w_i)^2 = s(s-1)\,\mathrm{Cov}_u\inparen{g_u(w_1),g_u(w_2)}-s\inparen{\EE_ug_u(w_1)}^2$. By Lemma~\ref{lem:ldpc-product}, $w_1$ has distribution $\omega_0$ for $u = 0$. Thus $F'(0) = s\sum_a\omega_0(a)g(a) = 0$, and by Taylor's theorem it suffices to bound $\mathrm{Cov}_u\inparen{g_u(w_1),g_u(w_2)}$ for $u\in[0,1]$.

We bound this covariance by Lemma~\ref{lem:ldpc-smooth}(1) with $\mu = \mu_u$. Note that $\mu_u$ is $\frac\delta2$-smooth. Indeed, $\mu_u$ is proportional to $\frac{\nu\otimes\nu}{\omega_0}\inparen{(1-u)\omega_0+u\omega}$. By Lemma~\ref{lem:ldpc-smooth}(3), $|\nu/\tau-1|\le\eta_{\delta_\tau}\le\eta_\delta$. In particular, $\frac{\nu\otimes\nu}{\omega_0}$ takes values between $(1-\eta_\delta)^2$ and $(1+\eta_\delta)^2\le2$. Note that $\delta\le\frac{\delta_\tau}2$, $\omega_0$ is $\delta$-smooth by Lemma~\ref{lem:ldpc-product}, and $\omega$ is $\delta$-smooth by assumption. Thus $(1-u)\omega_0+u\omega$ is $\delta$-smooth and $\mu_u$ is $\frac\delta2$-smooth. By Lemma~\ref{lem:ldpc-smooth}(1) and the Cauchy--Schwarz inequality,
\[
\mathrm{Cov}_u\inparen{g_u(w_1),g_u(w_2)}\;\le\;\eta_\delta\sum_a\mu_u(a)g_u(a)^2\;\le\;4\eta_\delta\sum_a\frac{\omega_0(a)g(a)^2}{1+ug(a)},
\]
where we used $\nu\otimes\nu\le2\omega_0$ and $\sum_a(\nu\otimes\nu)_a(1+ug(a))\ge\frac12$. Hence $F(1)-F(0)\le4s^2\eta_\delta\,\chi^2(g)$. Indeed, $\int_0^1\frac{1-u}{1+ug(a)}\,du\le1$, since $g\ge-1$. By Pinsker's inequality, $\chi^2(g)\le|\omega-\omega_0|_1^2/\min_a\omega_0(a)\le2\tau_{\min}^{-2}I(\omega)$. For $s$ large enough, $t\frac ns(F(1)-F(0))\le8ts\eta_\delta\tau_{\min}^{-2}\,nI(\omega)\le\frac14nI(\omega)$.

Finally, by Stirling's formula $\varrho_0(\omega)\ge e^{-nI_n(\omega)-O(\log n)}$, and if $|\omega_a/\omega_0(a)-1|\le\frac12$ for all $a$, then $\varrho_0(\omega)\ge c\,n^{-d_V/2}e^{-nI_n(\omega)}$ for some constant $c>0$. By Fact~\ref{fact:llt}, $\PR{S_\nu = n\tau_n}\ge c\,n^{-(q^m-1)/2}$. We bound $\PR{S_{\mu_1} = n\omega}$ by $1$. If $|\omega_a/\omega_0(a)-1|\le\frac12$ for all $a$, then $\mu_1$ lies in a compact set of $\frac\delta2$-smooth distributions of full support. By Fact~\ref{fact:llt}, $\PR{S_{\mu_1} = n\omega}\le Cn^{-(q^{2m}-1)/2}$. The claim follows from the bounds above, $n|I_n(\omega)-I(\omega)| = 2nD(\tau_n\,\|\,\tau) = o(1)$, $\varrho_0(\omega)^{1-t}\varrho(\omega)^{t} = \varrho_0(\omega)\inparen{\varrho(\omega)/\varrho_0(\omega)}^t$, $t\ge1$ and $q^{2m}-1-2(q^m-1) = d_V$.
\end{proof}

The space of functions on $\F_q^{2m}$ with mean $0$ under $\omega_0$ is $(W\otimes\one)\oplus(\one\otimes W)\oplus(W\otimes W)$, and $A_{\nu\otimes\nu}$ maps each of these three spaces to itself. Indeed, by Lemma~\ref{lem:ldpc-product}, $\pi_{\nu\otimes\nu} = \Eover r{\pi^{(r)}\otimes\pi^{(r)}}$. Taking the sum over the second coordinate of each of the two rows, and using that both marginals of each $\pi^{(r)}$ are $\tau$, we obtain $\inparen{\pi_{\nu\otimes\nu}-\omega_0\omega_0^{\top}}(g\otimes\one) = \inparen{(\pi_\nu-\tau\tau^{\top})g}\otimes\tau$ for $g\in W$. Also $\Sigma(\omega_0)(h\otimes\one) = (\Sigma(\tau)h)\otimes\tau$ for $h\in W$. Thus $A_{\nu\otimes\nu}(g\otimes\one) = (A_\nu g)\otimes\one$, and similarly $A_{\nu\otimes\nu}(\one\otimes g) = \one\otimes(A_\nu g)$. Finally, $A_{\nu\otimes\nu}$ maps $W\otimes W$ to itself, because $\sum_a(\pi^{(r)}g)(a) = \sum_b\tau(b)g(b) = 0$ for $g\in W$.

\begin{claim}\label{claim:ldpc-near}
The first sum in~\eqref{eq:ldpc-split} is $(C_s+o(1))(\E X)^2$.
\end{claim}

\begin{proof}
For $\omega\in\Omega_n$ let $y := \sqrt n(\omega-\omega_n)$, which lies in $V$. The points $y$ form a coset of $\frac1{\sqrt n}\Gamma_V$, where $\Gamma_V$ is the lattice of integer vectors on $\F_q^{2m}$ whose two marginals are zero.

We first show that for any $K_0$ and all $y$ with $|y|\le K_0$,
\begin{equation}\label{eq:ldpc-pointwise}
n^{d_V/2}\varrho_0(\omega) = c_0\,e^{-\frac12\chi^2(y)}+o(1),\qquad
n^{d_V/2}\varrho(\omega) = c_1\,e^{-\frac12\chi^2\circ(I+A)^{-1}(y)}+o(1),
\end{equation}
for some constants $c_0,c_1>0$, where the $o(1)$ terms do not depend on $y$ and $\chi^2\circ(I+A)^{-1}(y) := \sum_ay_a\inparen{(I+A)^{-1}(y/\omega_0)}(a)$. For $\varrho_0$, \eqref{eq:ldpc-pointwise} follows from Stirling's formula and $nI_n(\omega) = \frac12\chi^2(y)+o(1)$. For $\varrho$, we apply Fact~\ref{fact:llt} to $S_{\nu\otimes\nu}$ and to $S_\nu$. By Lemma~\ref{lem:ldpc-product}, $S_{\nu\otimes\nu}$ has expectation $n\omega_0$, and $n(\omega-\omega_0) = \sqrt ny+O(1)$. By Lemma~\ref{lem:ldpc-smooth}(2), $S_{\nu\otimes\nu}$ has covariance $n\Sigma(\omega_0)(I+A_{\nu\otimes\nu})$. Since $A_{\nu\otimes\nu}$ maps $W\otimes W$ to itself, on $V$ the inverse of the covariance is $\chi^2\circ(I+A)^{-1}$, and we obtain~\eqref{eq:ldpc-pointwise}.

Next, we compute $c_1/c_0$. Both marginal distributions of $\omega$ are $\tau_n$. Hence we have $\varrho_0(\omega) = \PR{T = n\omega}/\PR{T_\tau = n\tau_n}^2$, where $T$ and $T_\tau$ are $n$ times the types of matrices whose $n$ rows are independent samples of $\omega_0$ and of $\tau$, respectively. We apply Fact~\ref{fact:llt} to $S_{\nu\otimes\nu}$, $S_\nu$, $T$ and $T_\tau$. The determinants of the lattices of $S_{\nu\otimes\nu}$ and $S_\nu$ are $\iota_{2m}$ and $\iota_m$ times the determinants of the lattices of $T$ and $T_\tau$, respectively, and $\iota_{2m} = \iota_m^2$. By Lemma~\ref{lem:ldpc-smooth}(2), the covariances of $S_{\nu\otimes\nu}$, $S_\nu$, $T$ and $T_\tau$ are $n\Sigma(\omega_0)(I+A_{\nu\otimes\nu})$, $n\Sigma(\tau)(I+A_\nu)$, $n\Sigma(\omega_0)$ and $n\Sigma(\tau)$, respectively. Hence, $c_1/c_0 = \det'(I+A_{\nu\otimes\nu})^{-1/2}\det'(I+A_\nu)$, where $\det'$ is the determinant on the functions modulo constants. Since $A_{\nu\otimes\nu}$ coincides with $A_\nu$ on $W\otimes\one$ and on $\one\otimes W$, we have $\det'(I+A_{\nu\otimes\nu}) = \det'(I+A_\nu)^2\det_V(I+A)$. Thus $c_1/c_0 = \det_V(I+A)^{-1/2}$. We now compute $c_0$. Note that $\sum_{\omega\in\Omega_n}\varrho_0(\omega) = 1$, and by Stirling's formula and Pinsker's inequality the $\varrho_0(\omega)$ with $|y|>K_0$ sum to at most $\eps_{K_0}+o(1)$, where $\eps_{K_0}\to0$ as $K_0\to\infty$. By~\eqref{eq:ldpc-pointwise}, and recalling that the points $y$ form a coset of $\frac1{\sqrt n}\Gamma_V$, the sum of the $\varrho_0(\omega)$ with $|y|\le K_0$ tends to $\frac{c_0}{\det(\Gamma_V)}\int_{|y|\le K_0}e^{-\frac12\chi^2(y)}\,dy$. Thus $c_0\int_Ve^{-\frac12\chi^2(y)}\,dy = \det(\Gamma_V)$.

Finally, by Claim~\ref{claim:ldpc-far} (the first bound for the $\omega$ with $|\omega_a/\omega_0(a)-1|>\frac12$ for some $a$, and the second bound otherwise), the terms of the first sum in~\eqref{eq:ldpc-split} with $|y|>K_0$ sum to at most $(\eps_{K_0}+o(1))(\E X)^2$, and by~\eqref{eq:ldpc-pointwise}, as for the $\varrho_0(\omega)$ above, $(\E X)^{-2}$ times the sum of the terms with $|y|\le K_0$ tends to $\frac{c_0^{1-t}c_1^t}{\det(\Gamma_V)}\int_{|y|\le K_0}e^{-\frac12\inbrak{(1-t)\chi^2(y)+t\chi^2\circ(I+A)^{-1}(y)}}\,dy$. Hence $(\E X)^{-2}$ times the first sum tends to
\[
\frac{c_0^{1-t}c_1^t}{\det(\Gamma_V)}\int_Ve^{-\frac12\inbrak{(1-t)\chi^2(y)+t\chi^2\circ(I+A)^{-1}(y)}}\,dy
\;=\;\det\nolimits_V(I+A)^{-\frac t2}\det\nolimits_V\inparen{(1-t)I+t(I+A)^{-1}}^{-\frac12}.
\]
By Lemma~\ref{lem:ldpc-smooth}(2), $t\|A\|\le ts\eta_{\delta_\tau/2}<1$. Thus $(1-t)I+t(I+A)^{-1} = (I+A)^{-1}\inparen{I-(t-1)A}$ is positive definite and the right-hand side above is $C_s$.
\end{proof}

Combining~\eqref{eq:ldpc-split} with Claims~\ref{claim:ldpc-lowrank} and~\ref{claim:ldpc-near}, and noting that $\E Y = (1-o(1))\E X$, we have $\E{Y^2}\le(C_s+o(1))(\E Y)^2$.
Finally, $C_s\to1$ as $s\to\infty$, since the eigenvalues of $A$ are at most $s\eta_{\delta_\tau/2}$ in absolute value, as in the proof of Proposition~\ref{prop:ldpc-identity}, and $ts\eta_{\delta_\tau/2}\to0$.

\section*{AI Methodology}

Claude Opus 5 assisted in this work by (1) finding a counterexample to the average-radius list decoding version of the conjecture of \cite{GLMRSW}, which does not appear in this paper but informed our work greatly, and (2) helping to extend our result about random LDPC codes from constant probability to probability $1-o(1)$ (Sections~\ref{sec:ldpc-cycles} and~\ref{sec:ldpc-assembly} and Appendix~\ref{app:ldpc-cycles}). The author is responsible for the correctness of the arguments and the originality of this work.

\end{document}